\documentclass[11pt]{article}
\usepackage[letterpaper, margin=1in]{geometry}
\usepackage[utf8]{inputenc}
\usepackage{newtxtext}
\usepackage{xcolor}
\def\anonymous{0}      % 1 = anonymous
\def\PrintMode{0}%normally PrintMode = 0. If PrintMode = 1, extreme methods will be used to reduce the number of pages. Can also set fontsize = 10pt to further shrink pages...

\usepackage{tikz}
\usepackage{booktabs}
\usetikzlibrary{positioning}
\usepackage{pgffor}
\usetikzlibrary{decorations.pathreplacing}
\usepackage{expl3}
\usepackage[T1]{fontenc}
\usepackage{graphicx}
\usepackage{float}
\usepackage[lined,ruled,linesnumbered]{algorithm2e}
\usepackage{xspace}
\usepackage[utf8]{inputenc}
\usepackage[noadjust]{cite}
\usepackage{amsmath, amssymb, amsfonts, amsthm}
\usepackage{enumerate}
\usepackage{xcolor}
\usepackage[all]{xy}
\usepackage[linktocpage=true,pagebackref=true,colorlinks,linkcolor=magenta,citecolor=blue,bookmarks,bookmarksopen,bookmarksnumbered]{hyperref}
\usepackage{bm}
\usepackage{bbm}
\usepackage[normalem]{ulem}
\usepackage{paralist}
\usepackage{gensymb}
\usepackage{thmtools}
\usepackage{rotating}
\usepackage{mdframed}
\usepackage{multirow}
\usepackage{appendix}
\usepackage{tabularx}
\usepackage{verbatim}
\usepackage{mathrsfs}
\usepackage{indentfirst}
\usepackage{mleftright}
\usepackage[nottoc]{tocbibind}%add bib in toc
\usepackage{complexity}
\usepackage{tablefootnote}
\usepackage{epigraph}
\usepackage{autobreak}

\ifnum\PrintMode=1

\usepackage{savetrees}
\else
\fi

\graphicspath{{./figs/}}

\usepackage[capitalize]{cleveref}%cleveref must be loaded *last*

\newtheorem{theorem}{Theorem}[section]
\newtheorem{lemma}[theorem]{Lemma}
\newtheorem{proposition}[theorem]{Proposition}
\newtheorem{corollary}[theorem]{Corollary}

\theoremstyle{definition}
\newtheorem{definition}[theorem]{Definition}

\newtheorem{remark}[theorem]{Remark}

\renewcommand*\backref[1]{\ifx#1\relax \else (cit.~on p.~#1) \fi} %http://latex.org/forum/viewtopic.php?t=3670

\makeatletter
\def\moverlay{\mathpalette\mov@rlay}
\def\mov@rlay#1#2{\leavevmode\vtop{%
		\baselineskip\z@skip \lineskiplimit-\maxdimen
		\ialign{\hfil$\m@th#1##$\hfil\cr#2\crcr}}}
\newcommand{\charfusion}[3][\mathord]{
	#1{\ifx#1\mathop\vphantom{#2}\fi
		\mathpalette\mov@rlay{#2\cr#3}
	}
	\ifx#1\mathop\expandafter\displaylimits\fi}
\makeatother

\renewcommand{\emptyset}{\varnothing}

\newlang{\MCSP}{MCSP}
\newlang{\MFSP}{MFSP}
\newlang{\MKtP}{MKtP}
\newlang{\MKTP}{MKTP}
\newlang{\itrMCSP}{itrMCSP}
\newlang{\itrMKTP}{itrMKTP}
\newlang{\itrMINKT}{itrMINKT}
\newlang{\MINKT}{MINKT}
\newlang{\MINK}{MINK}
\newlang{\MINcKT}{MINcKT}
\newlang{\CMD}{CMD}
\newlang{\DCMD}{DCMD}
\newlang{\CGL}{CGL}
\newlang{\PARITY}{PARITY}
\newlang{\Empty}{\textsc{Empty}}
\newlang{\Avoid}{\textsc{Avoid}}
\newlang{\Sparsification}{\textsc{Sparsification}}
\newlang{\HamEst}{\mathsf{HammingEst}}
\newlang{\HamHit}{\mathsf{HammingHit}}
\newlang{\CktEval}{\textsc{Circuit-Eval}}
\newlang{\Hard}{\textsc{Hard}}
\newlang{\cHard}{\textsc{cHard}}
\newlang{\CAPP}{CAPP}
\newlang{\GapUNSAT}{GapUNSAT}
\newlang{\OV}{OV}
\renewlang{\PCP}{PCP}
\newlang{\PCPP}{PCPP}
\newclass{\Avg}{Avg}
\newclass{\ZPEXP}{ZPEXP}
\newclass{\DLOGTIME}{DLOGTIME}
\newclass{\ALOGTIME}{ALOGTIME}
\newclass{\ATIME}{ATIME}%alternating time
\newclass{\SZKA}{SZKA}
\newclass{\Laconic}{Laconic\text{-}}
\newclass{\APEPP}{APEPP}
\newclass{\SAPEPP}{SAPEPP}
\newclass{\TFSigma}{TF\Sigma}
\newclass{\NTIMEGUESS}{NTIMEGUESS}

\newlang{\Formula}{Formula}
\newlang{\THR}{THR}
\newlang{\EMAJ}{EMAJ}
\newlang{\MAJ}{MAJ}
\newlang{\SYM}{SYM}
\newlang{\DOR}{DOR}
\newlang{\ETHR}{ETHR}
\newlang{\Midbit}{Midbit}
\newlang{\LCS}{LCS}
\newlang{\TAUT}{TAUT}
\newlang{\Poly}{\text{-}Poly}

\newcommand{\N}{\mathbb{N}}

\newcounter{prob}
\newenvironment{prob}[2][]{%
\refstepcounter{prob}%
\mdfsetup{%
frametitle={%
\tikz[baseline=(current bounding box.east),outer sep=0pt]
\node[anchor=east,rectangle,fill=white]
{\strut Problem~\arabic{prob}:~#1};}}%
\mdfsetup{innertopmargin=0pt,linecolor=black,%
linewidth=1pt,topline=true,%
frametitleaboveskip=\dimexpr-\ht\strutbox\relax
}
\begin{mdframed}[]\relax%
\label{#2}}{\end{mdframed}}

\renewcommand{\epsilon}{\varepsilon}

\usepackage[draft, inline,marginclue,index]{fixme}  %Remove draft to erase all notes, remove in line to erase inline notes.
\usepackage{subfloat}
\usepackage{subcaption}
\fxsetup{theme=color,mode=multiuser}

\definecolor{color1}{RGB}{46,134,193}
\FXRegisterAuthor{mahdi}{amahdi}{\color{orange}Mahdi}

\newcommand\blfootnote[1]{
  \begingroup
  \renewcommand\thefootnote{}\footnote{#1}
  \addtocounter{footnote}{-1}
  \endgroup
}

\title{Subset Balancing and Generalized Subset Sum via Lattices}
\ifnum\anonymous=0
\author{
Yiming Gao$^{1}$ \quad
Yansong Feng$^{2,3}$ \quad
Honggang Hu$^{1,4}$ \quad 
Yanbin Pan$^{2,3}$}

\date{
}
\else
\author{Anonymous}
\date{}
\fi
\begin{document}

\maketitle

\blfootnote{\textsuperscript{1}School of Cyber Science and Technology, University of Science and Technology of China, and Anhui Province Key Laboratory of Digital Security, Hefei, China, \texttt{qw1234567@mail.ustc.edu.cn, hghu2005@ustc.edu.cn}}

\blfootnote{\textsuperscript{2}KLMM, Academy of Mathematics and Systems Science, Chinese Academy of Sciences, Beijing, China}
\blfootnote{\textsuperscript{3}School of Mathematical Sciences, University of Chinese Academy of Sciences, Beijing, China, \texttt{\{fengyansong,panyanbin\}@amss.ac.cn}}
\blfootnote{\textsuperscript{4}Hefei National Laboratory, Hefei, China}

\begin{abstract}
We study the \emph{Subset Balancing} problem: given $\mathbf{x} \in \mathbb{Z}^n$ and a coefficient set $C \subseteq \mathbb{Z}$, find a nonzero vector $\mathbf{c} \in C^n$ such that $\mathbf{c}\cdot\mathbf{x} = 0$. The standard meet-in-the-middle algorithm runs in time $\tilde{O}(|C|^{n/2})=\tilde{O}(2^{n\log |C|/2})$, and recent improvements (SODA~2022, Chen, Jin, Randolph, and Servedio; STOC~2026, Randolph and Węgrzycki) beyond this barrier apply mainly when $d$ is constant.

We give a reduction from Subset Balancing with $C = \{-d, \dots, d\}$ to a single instance of $\mathrm{SVP}_{\infty}$ in dimension $n+1$. Instantiating this reduction with the best currently known $\ell_\infty$-SVP algorithms yields a deterministic algorithm with running time $\tilde{O}((6\sqrt{2\pi e})^n) \approx \tilde{O}(2^{4.632n})$, and a randomized algorithm with running time $\tilde{O}(2^{2.443n})$ (here $\tilde{O}$ suppresses $\operatorname{poly}(n)$ factors). Crucially, our running times depend only on $n$ and are independent of $d$, thereby removing the $\log d$ factor from the exponent entirely. Even for moderate $d$, the randomized bound already beats meet-in-the-middle for all $d\ge 15$. We also show that for sufficiently large $d$, Subset Balancing is solvable in polynomial time. More generally, we extend the box constraint $[-d,d]^n$ to an arbitrary centrally symmetric convex body $K \subseteq \mathbb{R}^n$, obtaining a deterministic $\tilde{O}(2^{c_K n})$-time algorithm, where $c_K$ depends only on the shape of $K$.

We further study the \emph{Generalized Subset Sum} problem of finding $\mathbf{c} \in C^n$ such that $\mathbf{c} \cdot \mathbf{x} = \tau$.  For $C = \{-d, \dots, d\}$ or $C = \{-d,\dots,d\}\setminus\{0\}$, we reduce the worst-case problem to $\mathrm{CVP}_{\infty}$ in dimension $n+1$.  Although no general single-exponential time algorithm is known for exact $\mathrm{CVP}_{\infty}$, we observe that distances in our lattice take only integer values, so an approximate $\mathrm{CVP}_{\infty}$ oracle still suffices. Combining this with the approximate $\mathrm{CVP}_\infty$ algorithm yields a deterministic worst-case algorithm running in time $2^{O(n\log\log d)}$, already asymptotically faster than meet-in-the-middle. 

In the average-case setting, we demonstrate that for both coefficient sets the embedded $\mathrm{CVP}_{\infty}$ instance satisfies a bounded-distance promise with high probability, which allows us to remove the $\log\log d$ factor altogether. This yields a deterministic algorithm running in time $\tilde{O}((18\sqrt{2\pi e})^n) \approx \tilde{O}(2^{6.217n})$.
\end{abstract}

\newpage

\section{Introduction}
\label{sec:introduction}
The Subset Sum problem is a classic NP-complete problem, and whether it admits an exact algorithm faster than $\tilde{O}(2^{0.5n})$ (where $\tilde{O}$ suppresses $\operatorname{poly}(n)$ factors) remains a central open question. Despite extensive efforts, the Meet-in-the-Middle algorithm of Horowitz and Sahni, with running time $\tilde{O}(2^{0.5n})$, remains the benchmark for worst-case instances~\cite{horowitz1974computing}.

In the average-case setting, a breakthrough result by Howgrave-Graham and Joux~\cite{howgrave2010new} introduced the \emph{representation technique}, enabling algorithms that beat Meet-in-the-Middle for several variants. In particular, they obtained an $\tilde{O}(2^{0.337n})$-time algorithm under reasonable heuristic assumptions, which was later improved to $\tilde{O}(2^{0.283n})$~\cite{becker2011improved,bohme2011verbesserte,bonnetain2020improved}. Following this line of work, a sequence of results has focused on subset sum and related variants. In 2019, the representation technique was further applied to the Equal Subset Sum problem, yielding an $\tilde{O}(3^{0.488n})$ worst-case algorithm, improving upon the previous $\tilde{O}(3^{0.5n})$ Meet-in-the-Middle bound~\cite{mucha2019equal}. This was recently further improved to $\tilde{O}(3^{0.487n})$~\cite{randolph2025beating}.

\subsection{Generalized Subset Sum, Subset Balancing, and the Meet-in-the-Middle Barrier}

Chen, Jin, Randolph, and Servedio~\cite{chen2022average} (SODA~2022) studied a generalized subset sum problem in which $\mathbf{x}$ is sampled uniformly from $[0, M-1]^n$, and one seeks $\mathbf{c} \in C^n$ satisfying $\mathbf{c}\cdot\mathbf{x} = \tau$.

\begin{prob}[Generalized Subset Sum (GSS)]{prob:gss}
\textbf{Input.} An input bound $M$, a vector $\mathbf{x} = (x_1, \dots, x_n) \in [0, M-1]^n$, a set $C \subset \mathbb{Z}$ of allowed coefficients, and a target integer $\tau$. \\
\textbf{Output.} A vector $\mathbf{c} \in C^n$ such that $\mathbf{c}\cdot\mathbf{x} = \tau$, if one exists.
\end{prob}

This formulation captures, for example, Equal Subset Sum (ESS) when $C = \{-1,0,1\}$, and Partition when $C = \{-1,1\}$ with $\tau=0$. For constant-size coefficient sets $C$, for example $C=[-d:d] := \{-d,-d+1,\ldots,d\}$ or $C=[\pm d] := \{-d,\ldots,-1,1,\ldots,d\}$ where $d$ is a fixed integer, Chen et al.~\cite{chen2022average} obtained running times of the form $|C|^{(1/2 - \Omega(1/|C|))n}$ (with high success probability), together with sharp structural results about when solutions exist in the average-case setting. Rewriting this in base $2$, the improvement over meet-in-the-middle lies in reducing the constant multiplying $n \log d$ from $1 / 2$ to $1 / 2-\Omega(1 |C|)$, but the $\log d$ factor in the exponent is preserved. Moreover, the savings $\Omega(1/|C|)$ vanish as $d$ grows, so the improvement becomes negligible for large coefficient sets.

More recently, Randolph and W\k{e}grzycki~\cite{randolph2025beating} (STOC~2026) focused on the case $\tau = 0$, referred to as the Subset Balancing problem:

\begin{prob}[Subset Balancing]{prob:sb}
\textbf{Input.} A vector $\mathbf{x} \in \mathbb{Z}^n$ and a set $C \subset \mathbb{Z}$ of allowed coefficients. \\
\textbf{Output.} A vector $\mathbf{c} \in C^n \setminus \{\mathbf{0}\}$ such that $\mathbf{c}\cdot\mathbf{x} = 0$, if one exists.
\end{prob}

They obtained \emph{worst-case} algorithms that break the Meet-in-the-Middle barrier for a wide range of constant-size coefficient sets, $[-d:d]$ for $d\ge1$ and $[\pm d]$ for $d>2$. Their work extends representation-based methods to worst-case inputs via new tools such as coefficient shifting and compatibility certificates, achieving running time $\tilde{O}(|C|^{(1/2 - \varepsilon)n})$ for some constant $\varepsilon > 0$ depending on $C$. Also, their techniques are designed for fixed $d$ , as $d$ grows, both the improvement $\varepsilon$ and the algorithmic framework become increasingly unwieldy. For the specific cases computed in~\cite{randolph2025beating}, the savings $\varepsilon$ decrease from $ 0.022$ for $C=[-2: 2]$ to $0.005$ for $C=[\pm 3]$, consistent with the interpretation that the exponent remains $\Theta(n \log d)$.

These barrier-breaking results leave open a complementary algorithmic question: \emph{what happens when $d$ is large?} Intuitively, larger $d$ should make subset balancing “easier” because the feasible set $C^n$ grows rapidly. At the same time, Meet-in-the-Middle becomes less attractive, as its running time scales roughly as $(2d)^{n/2}$.

This motivates the following questions that guide our work:
\begin{enumerate}
  \item Can we formalize the intuition that large $d$ leads to easy (or even polynomial-time) solvability?
  \item Can we design algorithms that outperform Meet-in-the-Middle when $d$ is large, ideally with running time depending primarily on the dimension $n$ rather than on $|C|$? In particular, can this be achieved \emph{deterministically}?
\end{enumerate}

We answer these questions by taking a lattice-theoretic viewpoint.

\subsection{Our Contributions}

\paragraph{Subset Balancing Problem.} In this work, we first focus on the Subset Balancing Problem with
\(C = [-d:d]\) in the worst case. We consider the lattice consisting of all integer solutions $\mathbf{c}$ such that $\mathbf{c} \cdot \mathbf{x} = 0$, thereby reducing the task of finding $\mathbf{c} \in C^n$ to an instance of $\mathrm{SVP}_{\infty}$ (the Shortest Vector Problem in the $\ell_{\infty}$ norm). However, this lattice is not full rank. To leverage the state-of-the-art deterministic algorithms for $\mathrm{SVP}_{\infty}$ based on the framework of Dadush, Peikert, and Vempala~\cite{dadush2011enumerative} (FOCS~2011), we construct a full-rank lattice by a standard embedding. We further refine the analysis in the $\ell_\infty$ setting and obtain the explicit base \(6\sqrt{2\pi e}\) appearing in Theorem~\ref{thm:mainworstcase}. For the randomized setting, we rely on the SVP algorithm of Mukhopadhyay~\cite{mukhopadhyay2019faster}, which yields Theorem~\ref{thm:mainworstcase}.

\begin{theorem}[Algorithms for Worst-Case SBP with \texorpdfstring{$C=[-d:d]$}{C=[-d:d]}]
\label{thm:mainworstcase}
Given any $d,n \in \mathbb{N}$ and let $C = [-d:d]$, there is a deterministic algorithm for Subset Balancing Problems in running time
\[
\tilde{O}\!\left(\bigl(6\sqrt{2\pi e}\,\bigr)^n\right)\approx \tilde{O}(2^{4.632n}),
\]    
and a randomized algorithm in time
\[
\tilde{O}(2^{2.443n}).
\]
\end{theorem}
Our randomized algorithm is asymptotically faster than the Meet-in-the-Middle algorithm, which runs in time \(\tilde{O}((2d+1)^{n/2})\), whenever \(d \ge 15\). Thus, already for moderately large coefficient bounds, the lattice-based approach yields a genuine improvement over meet-in-the-middle, while also providing a clean deterministic alternative.

We further obtain two additional results. First, for sufficiently large $d$, we derive a polynomial-time algorithm by replacing the SVP oracle with the LLL algorithm~\cite{lenstra1982factoring}. Second, we generalize the coefficient set from the box $[-d:d]^n$ to an arbitrary centrally symmetric convex body $K \subseteq \mathbb{R}^n$ by employing deterministic single-exponential SVP algorithms in general norms~\cite{DadushVempala2013Mellipsoids}. We obtain a deterministic algorithm running in time $\tilde{O}(2^{c_K n})$, where the exponent $c_K$ depends only on the \emph{shape} of $K$ and not on its scale. This abstraction suggests that the coefficient set $C$ can be interpreted as a geometric constraint, which may be useful for studying other variants of subset balancing.

\paragraph{Generalized Subset Sum.} Using an analysis analogous to the one above, we reduce Generalized Subset Sum with $C=[-d:d]$ to $\mathrm{CVP}_{\infty}$ in the worst case. However, unlike $\mathrm{SVP}_{\infty}$, no general single exponential time algorithm is known for exact $\mathrm{CVP}_{\infty}$, and when $C=[\pm d]$ the short vector in the lattice may correspond to an invalid solution in the worst case.

We overcome these by two key observations.

The first observation is that in our lattice, distance from the target to any lattice point can take only integer values, which lets us replace an exact $\mathrm{CVP}_\infty$ oracle by an approximate one. Specifically, for the $C=[-d:d]$ cases, we only need an approximate $\text{CVP}_{\infty}$ oracle with approximation factor $\left(1+\frac{1}{d}-\varepsilon\right)$ for arbitrarily small $\varepsilon>0$. We apply the oracle proposed by Eisenbrand, H\"ahnle, and Niemeier~\cite{eisenbrand2011covering}, yielding the overall complexity $2^{O(n \log\log d)}$ in~\cref{thm:worstcasegss}.

When \(C=[\pm d]=[-d,-1]\cup[1,d]\), the second observation is that \(C^n\) can be viewed as a union of \(2^n\) boxes. For each box, we shift the target vector then call the approximation $\text{CVP}_{\infty}$ oracle. Multiplying by the \(2^n\) factor yields the same overall time $2^{O(n \log\log d)}$ in~\cref{thm:worstcasegss}.

\begin{theorem}[Algorithms for Worst-Case GSS with \texorpdfstring{$C=[-d:d]$ or $[\pm d]$}{C=[-d:d] or [\pm d]}]
\label{thm:worstcasegss}
Given any $d>2,n \in \mathbb{N}$ and let $C = [-d:d]$ or $[\pm d]$, there is a deterministic algorithm for Generalized Subset Sum Problem in running time
\[
2^{O(n\log\log d)}.
\]    
\end{theorem}

Finally we study the average cases as~\cite{chen2022average}, in which we remove the $\log\log d$ factor altogether. 

We show a bounded distance condition holds with high probability. Specifically, when $\mathbf{x}\sim[0:M-1]^n$ and $|\tau|=o(Mn)$, we prove that, with probability $1-e^{-\Omega(n)}$, the distance between the target vector and the lattice is at most $4$ times the shortest vector length. Notably, for $C=[\pm d]$, we have an analogous result, but with a larger failure probability $1-o_n(1)$ in the following \cref{thm:mainaveragecase}.

We obtain the following result.

\begin{theorem}[Algorithms for Average-Case GSS with \texorpdfstring{$C=[-d:d]$ or $[\pm d]$}{C=[-d:d] or [\pm d]}]
\label{thm:mainaveragecase}
Fix $d\in\mathbb{N}$ and let $C=[-d:d]$ or $[\pm d]$. There is a deterministic algorithm for average-case Generalized Subset Sum running in time
$$
\tilde{O}\!\left(\bigl(18\sqrt{2\pi e}\,\bigr)^n\right)\approx \tilde{O}(2^{6.217n}) .
$$
For any $M$ and $\tau$ with $|\tau|=o(nM)$, the algorithm succeeds on $(M,\tau,\mathbf{x})$ with probability at least $1-e^{-\Omega(n)}$ when $C=[-d:d]$ and with probability at least $1-o_n(1)$ when $C=[\pm d]$ (over $\mathbf{x}\sim[0:M-1]^n$).
\end{theorem}

\begin{table}[t]
\centering
\caption{Comparison of algorithms for Subset Balancing (SBP) and Generalized
Subset Sum (GSS) with $C = [-d:d]$.  The column ``Exp.\ in $d$'' summarizes how
the \emph{exponent} of the running time scales with~$d$.  Here $\tilde{O}$
suppresses $\operatorname{poly}(n)$ factors.  Analogous results for
$C = [\pm d]$ are stated in the text.}
\label{tab:comparison}
\renewcommand{\arraystretch}{1.4}
\small
\begin{tabular}{@{} cl cc c c @{}}
\toprule
\textbf{Problem}
  & \textbf{Reference}
  & \textbf{Setting}
  & \textbf{Type}
  & \textbf{Running Time}
  & \textbf{Exp.\ in $d$} \\
\midrule
\multirow{5}{*}{\textbf{SBP}}
  & MitM~\cite{horowitz1974computing}
  & worst & det
  & $\tilde{O}\!\bigl(|C|^{n/2}\bigr)$
  & $\Theta(n\log d)$ \\[2pt]
  & \cite{randolph2025beating}
  & worst & rand
  & $\tilde{O}\!\bigl(|C|^{(1/2-\varepsilon)n}\bigr)^{\dagger}$
  & $\Theta(n\log d)$ \\[2pt]
  & \textbf{Ours}
  & worst & det
  & $\tilde{O}(2^{4.632n})$
  & $O(n)$ \\[2pt]
  & \textbf{Ours}
  & worst & rand
  & $\tilde{O}(2^{2.443n})$
  & $O(n)$ \\[2pt]
  & \textbf{Ours}
  & worst & det
  & $\operatorname{poly}(n,\log d)^{\ddagger}$
  & --- \\
\midrule
\multirow{4}{*}{\textbf{GSS}}
  & MitM~\cite{horowitz1974computing}
  & worst & det
  & $\tilde{O}\!\bigl(|C|^{n/2}\bigr)$
  & $\Theta(n\log d)$ \\[2pt]
  & \cite{chen2022average}
  & avg & rand
  & $\tilde{O}\!\bigl(|C|^{(1/2-\varepsilon)n}\bigr)^{\dagger}$
  & $\Theta(n\log d)$ \\[2pt]
  & \textbf{Ours}
  & worst & det
  & $2^{O(n\log\log d)}$
  & $O(n\log\log d)$ \\[2pt]
  & \textbf{Ours}
  & avg & det
  & $\tilde{O}(2^{6.217n})$
  & $O(n)$ \\
\bottomrule
\end{tabular}

\medskip
\raggedright\footnotesize
$^{\dagger}$\;$\varepsilon = \varepsilon(|C|) > 0$ depends on $|C|$ and
satisfies $\varepsilon \to 0$ as $d \to \infty$.\\[1pt]
$^{\ddagger}$\;Requires $d$ to be sufficiently large;
see Theorem~\ref{thm:sbppoly}.
\end{table}

A key conceptual difference from representation-based algorithms (including~\cite{chen2022average} and \cite{randolph2025beating}) is that our algorithm is deterministic and randomness appears only in the input distribution $\mathbf{x}\sim[0:M-1]^n$.

Our results are complementary to the barrier-breaking framework of
\cite{chen2022average,randolph2025beating}.
Representation-based algorithms excel in the regime where $d$ is a small constant, achieving exponents of the form $(1 / 2-\varepsilon) \cdot n \log |C|$. In contrast, our lattice-based approach operates in a fundamentally different regime: by reducing the problem to lattice problems, we remove the $\log d$ factor from the exponent altogether. 
This makes our algorithms particularly well-suited to the parameter regime where \(d\) is large, and also enables clean deterministic guarantees in settings where randomized hashing is undesirable.

\section{Preliminaries}
\label{sec:preliminaries}
We use $\mathbb{Z}$, $\mathbb{Q}$, and $\mathbb{R}$ to denote the ring of integers, the field of rationals, and the field of real numbers, respectively. Vectors are denoted by lowercase bold letters (e.g., $\mathbf{v}$) and are treated as row vectors unless otherwise specified. Matrices are denoted by uppercase bold letters (e.g., $\mathbf{A}$). 

For any vector $\mathbf{v} \in \mathbb{R}^m$ and $p \in [1, \infty]$, the $\ell_p$ norm of $\mathbf{v}$ is defined as
\begin{equation*}
\|\mathbf{v}\|_p =
\begin{cases}
\left(\sum_{i=1}^m |v_i|^p\right)^{1/p}, & 1 \le p < \infty,\\[1mm]
\max_{1 \le i \le m} |v_i|, & p = \infty.
\end{cases}
\end{equation*}
We write $\|\mathbf{v}\|$ for the Euclidean norm $\|\mathbf{v}\|_2$. 

For asymptotic notation, we use $O(\cdot)$, $\Theta(\cdot)$, and $\operatorname{polylog}(\cdot)$ in the standard way. We write $\tilde{O}(x)$ to denote $O(x \cdot \operatorname{polylog} x)$, suppressing factors polynomial in the logarithm of the input.

\paragraph{Parameters Assumptions}

We assume that all integer inputs ($d,x_i,M$) are bounded by $2^{n^{O(1)}}$ as in~\cite{randolph2025beating}. 
This is without loss of generality: integers of magnitude $2^{n^{\omega(1)}}$ are already infeasible to read or operate on in polynomial time in any realistic model. 
Notably, in~\cref{sec:worstcasealg,sec:worstcasegss} we can suppose $d$ as any integer even a function of $n$.

For the Generalized Subset Sum in Section~\ref{sec:averagecasealg}, we consider the same setting as~\cite{chen2022average} that $|\tau| = o(Mn)$, $d$ is a constant, and $M\geq 4^{n}$ otherwise we could apply standard dynamic programming technique.

\subsection{Lattices}
\label{sec:lattice}

A \emph{lattice} is a discrete additive subgroup of $\mathbb{R}^m$, where $m \in \mathbb{N}$. An equivalent definition is given below.

\begin{definition}[Lattice]
Let $\mathbf{v}_1, \mathbf{v}_2, \dots, \mathbf{v}_n \in \mathbb{R}^m$ be $n$ linearly independent vectors with $n \le m$. The \emph{lattice} $\mathcal{L}$ spanned by $\{\mathbf{v}_1, \mathbf{v}_2, \dots, \mathbf{v}_n\}$ is
\[
\mathcal{L} = \left\{ \mathbf{v} \in \mathbb{R}^m \ \middle| \ \mathbf{v} = \sum_{i=1}^{n} a_i \mathbf{v}_i,\; a_i \in \mathbb{Z} \right\}.
\]
\end{definition}

The integer $n$ is called the \emph{rank} of $\mathcal{L}$, while $m$ is its \emph{dimension}. The lattice $\mathcal{L}$ is \emph{full-rank} if $n = m$.

A lattice can be represented by a basis matrix $\mathbf{B} \in \mathbb{R}^{n \times m}$ whose rows are the basis vectors $\mathbf{v}_i$. The \emph{determinant} of $\mathcal{L}$ is
\[
\det(\mathcal{L}) = \sqrt{\det(\mathbf{B}\mathbf{B}^t)},
\]
where $\mathbf{B}^t$ denotes the transpose of $\mathbf{B}$. If $\mathcal{L}$ is full-rank, this simplifies to
\[
\det(\mathcal{L}) = \left| \det(\mathbf{B}) \right|.
\]

\begin{definition}[Successive Minima]
Let $p \in [1, \infty]$. Let $\mathcal{L} \subseteq \mathbb{R}^m$ be a lattice of rank $n$. For $1 \le i \le n$, the $i$-th \emph{successive minimum} of $\mathcal{L}$ in the $\ell_p$ norm, denoted $\lambda_i^{(p)}(\mathcal{L})$, is
\[
\lambda_i^{(p)}(\mathcal{L}) = \inf \left\{ r > 0 \ \middle| \ \dim\big(\mathrm{span}(\mathcal{L} \cap B_p(\mathbf{0}, r))\big) \ge i \right\},
\]
where $B_p(\mathbf{0}, r) = \{ \mathbf{x} \in \mathbb{R}^m : \|\mathbf{x}\|_p \le r \}$.
\end{definition}

\begin{lemma}[Minkowski's Theorem]
\label{lemma:minkowski}
Let $\mathcal{L} \subseteq \mathbb{R}^m$ be a lattice of rank $n$. Then
\[
\lambda_1^{(\infty)}(\mathcal{L}) \le \det(\mathcal{L})^{1/n}.
\]
\end{lemma}

As a corollary, we obtain the following bound for the $\ell_2$ norm.

\begin{corollary}[Minkowski's Theorem]
Let $\mathcal{L} \subseteq \mathbb{R}^m$ be a lattice of rank $n$. Then
\[
\lambda_1^{(2)}(\mathcal{L}) \le \sqrt{n} \cdot \det(\mathcal{L})^{1/n}.
\]
\end{corollary}

\begin{definition}[Shortest Vector Problem (SVP$_p$)]
Let $p \in [1, \infty]$. Given a lattice $\mathcal{L}$, the \emph{Shortest Vector Problem in the $\ell_p$ norm (SVP$_p$)} asks to find a nonzero vector $\mathbf{v} \in \mathcal{L}$ such that
\[
\|\mathbf{v}\|_p = \min_{\mathbf{w} \in \mathcal{L} \setminus \{\mathbf{0}\}} \|\mathbf{w}\|_p.
\]
\end{definition}

It was first shown to be NP-hard in the $\ell_\infty$ norm by van Emde Boas~\cite{van1981another}, who also conjectured that the same hardness should hold for the $\ell_2$ norm. Nearly two decades later, Ajtai proved that SVP in the $\ell_2$ norm is NP-hard under randomized reductions, unless $\mathrm{NP} \subseteq \mathrm{RP}$~\cite{ajtai1998shortest}. More recently, Hair and Sahai showed that $\mathrm{SVP}_p$ is NP-hard to approximate within a factor of $2^{\log^{1-\varepsilon} n}$ for all constants $\varepsilon > 0$ and $p > 2$, under standard deterministic Karp reductions~\cite{hair2025svp}. Nevertheless, the LLL algorithm~\cite{lenstra1982factoring} computes a relatively short vector in polynomial time.

\begin{lemma}[LLL Basis Reduction]
\label{lemma:LLL}
Let $\mathcal{L} \subseteq \mathbb{R}^m$ be a lattice of rank $n$. The LLL algorithm returns, in polynomial time, a nonzero vector $\mathbf{v} \in \mathcal{L}$ satisfying
\[
\|\mathbf{v}\|_2 \le 2^{\frac{n-1}{4}} \det(\mathcal{L})^{1/n}.
\]
\end{lemma}

Next, we will introduce the Closest Vector Problem, which was proven to be NP-hard for every $1 \leq p \leq \infty$~\cite{van1981another}.
\begin{definition}[Closest Vector Problem (CVP$_p$)]
Let $p \in [1, \infty]$. Given a lattice $\mathcal{L} \subset \mathbb{R}^m$ and a target vector $\mathbf{t} \in \mathbb{R}^m$, the \emph{Closest Vector Problem in the $\ell_p$ norm (CVP$_p$)} asks to find
\[
\mathbf{v} \in \arg\min_{\mathbf{w} \in \mathcal{L}} \|\mathbf{t} - \mathbf{w}\|_p.
\]
\end{definition}

\begin{definition}[$\gamma$-Approximate Closest Vector Problem ($\gamma$-CVP$_p$)]
Let $p \in [1, \infty], \gamma>1$. Given a lattice $\mathcal{L} \subset \mathbb{R}^m$ and a target vector $\mathbf{t} \in \mathbb{R}^m$, the \emph{$\gamma$-Approximate Closest Vector Problem in the $\ell_p$ norm ($\gamma$-CVP$_p$)} asks to find $\mathbf{v}\in \mathcal{L}$ satisfying
\[
\|\mathbf{v} - \mathbf{t}\|_p\le \gamma \min_{\mathbf{w} \in \mathcal{L}} \|\mathbf{t} - \mathbf{w}\|_p.
\]
\end{definition}

We are also interested in lattices associated with an integer vector $\mathbf{x} \in \mathbb{Z}^n$. 
\begin{proposition}
For $\mathbf{x}=(x_1,\ldots,x_n)\ne \mathbf{0}$, the set of integer vectors 
$\mathbf{c}=(c_1,\ldots,c_n)$ satisfying
\[
\sum_{i=1}^n c_i x_i = 0
\]
forms a sublattice of $\mathbb{Z}^n$, denoted $\mathcal{L}^{\perp}_{\mathbf{x}}$. The lattice $\mathcal{L}^{\perp}_{\mathbf{x}}$ has rank $n-1$, and its determinant is
\[
\det(\mathcal{L}^{\perp}_{\mathbf{x}}) = \frac{\|\mathbf{x}\|_2}{\gcd(x_1,\ldots,x_n)},
\]
where $\|\mathbf{x}\|_2$ is the Euclidean norm of $\mathbf{x}$.
    
\end{proposition}

\begin{proof}
Consider the linear map $\varphi : \mathbb{Z}^n \to \mathbb{Z}$ defined by
\[
\varphi(\mathbf{c}) = \sum_{i=1}^n c_i x_i = \mathbf{c}\cdot\mathbf{x}.
\]
Then $\mathcal{L}^{\perp}_{\mathbf{x}} = \ker(\varphi)$, so it is a sublattice of $\mathbb{Z}^n$. Since $\varphi$ has rank $1$ whenever $\mathbf{x} \neq \mathbf{0}$, it follows that $\ker(\varphi)$ has rank $n-1$.

Let $g = \gcd(x_1,\ldots,x_n)$ and write $\mathbf{x} = g \mathbf{x}'$, where $\mathbf{x}' \in \mathbb{Z}^n$ is primitive (i.e., $\gcd(x_1',\ldots,x_n') = 1$). Then
\[
\mathcal{L}^{\perp}_{\mathbf{x}} = \{ \mathbf{c} \in \mathbb{Z}^n : \mathbf{c}\cdot\mathbf{x}' = 0 \}.
\]

Since $\mathbf{x}'$ is primitive, there exists a unimodular matrix $U \in \mathrm{GL}_n(\mathbb{Z})$ whose last row is $\mathbf{x}'$. Applying $U$ yields a lattice isomorphism of $\mathbb{Z}^n$, under which $\mathcal{L}^{\perp}_{\mathbf{x}}$ is mapped to
\[
\{ \mathbf{y} \in \mathbb{Z}^n : y_n = 0 \},
\]
with scaling in the orthogonal direction determined by $\|\mathbf{x}'\|_2$.

More concretely, the covolume (determinant) of $\mathcal{L}^{\perp}_{\mathbf{x}}$ equals the Euclidean norm of the primitive normal vector $\mathbf{x}'$, namely
\[
\det(\mathcal{L}^{\perp}_{\mathbf{x}}) = \|\mathbf{x}/g\|_2.
\]
We obtain
\[
\det(\mathcal{L}^{\perp}_{\mathbf{x}}) = \left\| \frac{\mathbf{x}}{g} \right\|_2 = \frac{\|\mathbf{x}\|_2}{g}.
\]

Therefore,
\[
\det(\mathcal{L}^{\perp}_{\mathbf{x}}) = \frac{\|\mathbf{x}\|_2}{\gcd(x_1,\ldots,x_n)}.
\qedhere
\]
\end{proof}

\subsection{Convex Geometry}

Let $B_2^n := \{\mathbf{x}\in\mathbb{R}^n : \|\mathbf{x}\|_2 \le 1\}$ and $B_\infty^n := \{\mathbf{x}\in\mathbb{R}^n : \|\mathbf{x}\|_\infty \le 1\}$. For $r>0$, we have $r B_\infty^n \subseteq r\sqrt{n}\, B_2^n$ since $\|\mathbf{x}\|_2 \le \sqrt{n}\|\mathbf{x}\|_\infty$.

\begin{definition}[Convex Body]
A set $K \subseteq \mathbb{R}^n$ is a \emph{convex body} if it is convex, compact, and full-dimensional. It is \emph{centrally symmetric} if $K = -K$.
\end{definition}

For sets $A, B \subseteq \mathbb{R}^n$, their Minkowski sum is $A + B = \{\mathbf{x} + \mathbf{y} : \mathbf{x} \in A, \mathbf{y} \in B\}$. For $\mathbf{t} \in \mathbb{R}^n$, we write $\mathbf{t} + A = \{\mathbf{t}\} + A$.

For a convex body $K \subseteq \mathbb{R}^n$ and a lattice $L \subseteq \mathbb{R}^n$, define
\[
G(K,L) = \max_{\mathbf{x} \in \mathbb{R}^n} |(K+\mathbf{x}) \cap L|,
\]
the maximum number of lattice points in a translate of $K$. Define the covering number
\[
N(A,B) = \min \{|\Lambda|: \Lambda \subseteq \mathbb{R}^n, A \subseteq B + \Lambda\},
\]
which is the minimum number of translates of $B$ required to cover $A$.

\begin{definition}[Gauge Function] 
For a convex body $K$ containing the origin ($\mathbf{0} \in K$), the \emph{gauge function} (or \emph{Minkowski functional}) is defined as
\begin{equation*}
\|\mathbf{x}\|_K = \inf \{r \ge 0 : \mathbf{x} \in rK \}, \quad \mathbf{x} \in \mathbb{R}^n.
\end{equation*}
The functional $\|\cdot\|_K$ is a seminorm. If $K$ is centrally symmetric, then $\|\cdot\|_K$ defines a norm. In particular, for the $\ell_p$ unit ball $B_p^n = \{\mathbf{x} \in \mathbb{R}^n : \|\mathbf{x}\|_p \le 1\}$, we have $\|\mathbf{x}\|_{B_p^n} = \|\mathbf{x}\|_p$.
\end{definition}

For a positive definite matrix $\mathbf{A} \in \mathbb{R}^{n \times n}$, define
\begin{equation*}
\langle \mathbf{x}, \mathbf{y} \rangle_{\mathbf{A}} = \mathbf{x}\mathbf{A}\mathbf{y}^t, \quad 
\|\mathbf{x}\|_{\mathbf{A}} = \sqrt{\mathbf{x}\mathbf{A}\mathbf{x}^t}.
\end{equation*}

\begin{definition}[Ellipsoid]
For a center $\mathbf{a} \in \mathbb{R}^n$, the \emph{ellipsoid} $E(\mathbf{A}, \mathbf{a})$ is defined as
\begin{equation*}
E(\mathbf{A}, \mathbf{a}) = \{\mathbf{x} \in \mathbb{R}^n : \|\mathbf{x} - \mathbf{a}\|_{\mathbf{A}} \le 1\}.
\end{equation*}
\end{definition}

We denote $E(\mathbf{A}) = E(\mathbf{A}, \mathbf{0})$. Note that $\|\mathbf{x}\|_{\mathbf{A}} = \|\mathbf{x}\|_{E(\mathbf{A})}$. The volume of an ellipsoid satisfies
\begin{equation*}
\operatorname{vol}(E(\mathbf{A}, \mathbf{a})) = \operatorname{vol}(E(\mathbf{A})) = \operatorname{vol}(B_2^n) \cdot \sqrt{\det(\mathbf{A}^{-1})}.
\end{equation*}
Moreover, $E(\mathbf{A})^* = E(\mathbf{A}^{-1})$.

\subsection{SVP and CVP algorithms}

We restate the algorithms in \cite{dadush2011enumerative,micciancio2010deterministic} and give an explicit bound for deterministic SVP in the $\ell_{\infty}$ norm.

\begin{lemma}[Ellipsoid-Enum \cite{dadush2011enumerative}, Proposition 4.1;\cite{micciancio2010deterministic}]
\label{lem:Ellipsoid-Enum}
There is an algorithm $\mathrm{Ellipsoid\mbox{-}Enum}$ that, given any positive definite $\mathbf{A} \in \mathbb{Q}^{n \times n}$, any basis $\mathbf{B}$ of an $n$-dimensional lattice $L \subseteq \mathbb{R}^{n}$, and any $\mathbf{t} \in \mathbb{R}^n$, computes the set $L \cap (E(\mathbf{A})+\mathbf{t})$ in deterministic time
\[
\tilde{O}(2^{2n} + 2^n \cdot |L \cap (E(\mathbf{A})+\mathbf{t})|).
\]
\end{lemma}
\begin{proof}
By an invertible linear change of variables we may assume the ellipsoid is the Euclidean unit ball.

Using \cite{micciancio2010deterministic}, we can deterministically compute the Voronoi cell
of \(L\), equivalently the set \(\mathcal V\subseteq L\setminus\{0\}\) of Voronoi relevant vectors,
in time \(\tilde O(2^{2n})\), with \(|\mathcal V|=\tilde O(2^n)\). MV10 also gives a deterministic
algorithm to solve CVP in time \(\tilde O(2^{2n})\); let \(c_0\in L\) be a closest lattice vector to \(t\).
If \(\|c_0-t\|_2>1\) then \(S=\emptyset\); otherwise \(c_0\in S\).

Now define the neighbor graph \(G=(L,E)\) where \(x\) is adjacent to \(x+v\) for every
\(v\in\mathcal V\). We enumerate \(S\) by running BFS in the induced subgraph on \(S\), starting from
\(c_0\): when a node \(x\in S\) is popped, we test all \(x+v\) (\(v\in\mathcal V\)) and enqueue those
with \(\|x+v-t\|_2\le 1\) not seen before.

Correctness follows from connectivity of \(G[S]\). If \(x\in S\) is not a closest vector to \(t\), then
\(t-x\) lies outside the Voronoi cell at the origin, whose facets are defined by \(\mathcal V\). Hence
there exists \(v\in\mathcal V\) such that \(\|t-(x+v)\|_2<\|t-x\|_2\). In particular, since
\(\|t-x\|_2\le 1\), we also have \(x+v\in S\). Repeating yields a strictly decreasing sequence of distances
that must terminate at a closest vector, i.e.\ at \(c_0\). Reversing this sequence gives a path in \(G[S]\)
from \(c_0\) to \(x\). Thus BFS starting at \(c_0\) visits all of \(S\).

Running time: preprocessing (Voronoi cell) and the one CVP call cost \(\tilde O(2^{2n})\). In the BFS,
each visited point checks \(|\mathcal V|=\tilde O(2^n)\) neighbors, so the traversal costs
\(\tilde O(2^n\cdot |S|)\). Therefore the total deterministic time is
\[
\tilde O\bigl(2^{2n}+2^n\cdot |L\cap(B_2^n+t)|\bigr)
=
\tilde O\bigl(2^{2n}+2^n\cdot |L\cap(E(A)+t)|\bigr),
\]
as claimed.
\end{proof}

\begin{algorithm}[t]
\caption{$\mathrm{SVP}_\infty(L,\varepsilon)$: Deterministic Search via $\mathrm{Ellipsoid\mbox{-}Enum}$}
\label{alg:svp-infty}
\KwIn{A full-rank lattice $L\subseteq \mathbb{R}^n$ and $\varepsilon\in(0,1)$.}
\KwOut{A vector $v\in L\setminus\{0\}$ with $\|v\|_\infty=\lambda_1^{(\infty)}(L)$.}

\BlankLine
Compute deterministically $z_2\in L\setminus\{0\}$ such that $\|z_2\|_2=\lambda_1^{(2)}(L)$ (using \cite{micciancio2010deterministic}).\;
Set $d \leftarrow \|z_2\|_2/\sqrt{n}$.\;

\BlankLine
\While{true}{
Call Lemma~\ref{lem:Ellipsoid-Enum} to compute
\[
U \leftarrow L \cap \bigl(d\sqrt{n}\,B_2^n\bigr),
\]
$S \leftarrow \{\,y\in U : \|y\|_\infty \le d\,\}$\;
    
\If{$S\setminus\{0\}\neq \emptyset$}{
\Return $v \in \arg\min_{y\in S\setminus\{0\}} \|y\|_\infty$\;
}
$d \leftarrow (1+\varepsilon)\,d$\;
}
\end{algorithm}

\begin{theorem}[SVP in $\ell_{\infty}$]
\label{thm:svp-infty}
There is a deterministic algorithm for SVP in the $\ell_{\infty}$ norm on any full-rank lattice in $\mathbb{R}^{n}$ running in time
\[
\tilde{O}\!\left(\bigl(6\sqrt{2\pi e}\,\bigr)^n\right) \approx \tilde{O}(2^{4.632n}).
\]
\end{theorem}

\begin{proof}\ \vspace{-12pt}  %hack

\paragraph{Correctness of Algorithm \ref{alg:svp-infty}.}
Since \(\|y\|_\infty \ge \|y\|_2/\sqrt{n}\) for all \(y\), we have
\[
d_0=\frac{\lambda_1^{(2)}(L)}{\sqrt{n}} \;\le\; \lambda_1^{(\infty)}(L),
\]
so the loop starts below or at the target. Eventually \(d\ge \lambda_1^{(\infty)}(L)\), hence \(L\cap dB_\infty^n\) contains a nonzero vector and the algorithm terminates. On the last iteration, by construction \(S = L\cap dB_\infty^n\), so minimizing \(\|\cdot\|_\infty\) over \(S\setminus\{0\}\) returns an \(\ell_\infty\)-shortest nonzero vector.
    
\paragraph{Running time of Algorithm \ref{alg:svp-infty}.}
Because \(\lambda_1^{(\infty)}(L)\le \lambda_1^{(2)}(L)\), we have \(\lambda_1^{(\infty)}(L)/d_0 \le \sqrt{n}\), hence the number of iterations is at most
\[
I \;=\; \left\lceil \log_{1+\varepsilon} \sqrt{n} \right\rceil.
\]
Each iteration calls Ellipsoid-Enum once with output size \(|U| = |L\cap rB_2^n|\), so it costs
\[
\tilde{O}\!\left(2^{2n} + 2^n \cdot |L\cap rB_2^n|\right).
\]
Let \(d^\star\) be the first radius where the algorithm terminates, and \(r^\star=d^\star\sqrt{n}\).
Then the total time is
\[
\tilde{O}\!\left(I\cdot (2^{2n}+ 2^n \cdot |L\cap r^\star B_2^n|) \right),
\]
dominated by the final enumeration.

Moreover, one can upper bound the final output size by
\[
|L\cap r^\star B_2^n|
\;\le\; G(d^*\sqrt{n}B_2^n,L)
\;\le\;
N(\sqrt{n}B_2^n, B_\infty^n)\cdot G(d^\star B_\infty^n,L).
\]
Using the volumetric covering bound by Rogers and Zong \cite{rogers1997covering}
\[
N(\sqrt{n}B_2^n, B_\infty^n)
\;\le\;
\frac{\operatorname{vol}(\sqrt{n}B_2^n-B_\infty^n)}{\operatorname{vol}(B_\infty^n)} \vartheta(B_\infty^n).
\]
Since $B_\infty^n \subseteq\sqrt{n}B_2^n$ and $\vartheta(B_\infty^n)$ is polynomial in $n$, thus
\[
N(\sqrt{n}B_2^n, B_\infty^n)
\;\le\;
\tilde{O}\left(\frac{\operatorname{vol}(2\sqrt{n}B_2^n)}{\operatorname{vol}(B_\infty^n)}\right)=
\tilde{O}(\operatorname{vol}(\sqrt{n}B_2^n))
=\tilde{O}((\sqrt{2\pi e})^n)).
\]
By Lemma 4.3 in \cite{dadush2011enumerative}  
\[
G(d^\star B_\infty^n,L)
\;\le\;
\left(1+\frac{2d^\star}{\lambda_1^{(\infty)}(L)}\right)^n,
\]
together with \(d^\star < (1+\varepsilon)\lambda_1^{(\infty)}(L)\) (by minimality of \(d^\star\) in the \((1+\varepsilon)\)-search),
we get
\[
G(d^\star B_\infty^n,L)
\;\le\;
(3+2\varepsilon)^n,
\]
hence
\[
|L\cap r^\star B_2^n|
\;\le\;
\tilde{O}( (\sqrt{2\pi e})^n \cdot (3+2\varepsilon)^n).
\]
Plugging into Lemma~\ref{lem:Ellipsoid-Enum}, the overall time is
\[
\tilde{O}\!\left(I\cdot (2^{2n}+ 2^n \cdot \bigl(\sqrt{2\pi e}\,(3+2\varepsilon)\bigr)^n) \right)
.
\]
By setting $\varepsilon=1/n$, $I=\log_{1+\varepsilon} \sqrt{n}$ is still polynomial in $n$ and $(3+2\varepsilon)^n=\tilde{O}(3^n)$, we obtain a deterministic SVP algorithm for $\ell_{\infty}$ in
\[
\tilde{O}\!\left(\bigl(6\sqrt{2\pi e}\,\bigr)^n\right)\approx \tilde{O}(2^{4.632n}).
\qedhere
\]
\end{proof}

For the general norm cases, Dadush and Vempala \cite{DadushVempala2013Mellipsoids} presented a fully deterministic single exponential SVP algorithm in any norm.
\begin{theorem}[Theorem 1.8 in \cite{DadushVempala2013Mellipsoids}]
\label{thm:svpanynorm}
Given a basis for a full rank lattice $L$ and a well-centered convex body $K$, both in $\mathbb{R}^n$, the shortest vector in $L$ under the norm $\|\cdot\|_K$ can be found deterministically, using $2^{O(n)}$ time and space.
\end{theorem}

They also present a exact $2^{O(n)}$ CVP algorithm when the target point is sufficiently close to the lattice.

\begin{theorem}[Theorem 1.9 in \cite{DadushVempala2013Mellipsoids}]
\label{thm:cvpanynorm}
Given a basis for a lattice $L$, any well-centered $n$ dimensional convex body $K$, and a query point $x$ in $\mathbb{R}^n$, the closest vector in $L$ to $x$ in the norm $\|.\|_K$ defined by $K$ can be computed deterministically, using $(2+\gamma)^{O(n)}$ time and space, provided that the minimum distance is at most $\gamma$ times the length of the shortest nonzero vector of $L$ under $\|\cdot\|_K$.

\end{theorem}

However, it remains open to solve the CVP in time $2^{O(n)}$ with no assumptions on the minimum distance. Even the special case of the CVP under any norm other than the Euclidean norm is open.

\section{Worst Case Algorithms for Subset Balancing Problems}
\label{sec:worstcasealg}
\subsection{Subset Balancing Problems with $C=[-d:d]$}

Given $\mathbf{x}=(x_1,\ldots,x_n)$, the goal is to find a nonzero vector $\mathbf{c} \in C^n$ such that $\mathbf{c} \cdot \mathbf{x} = 0$. In the following discussion, we assume by default $\mathbf{x}\ne\mathbf{0}$, otherwise the problem is trivial. In this section, we treat $d$ as a parameter rather than a constant in the setting of \cite{randolph2025beating,chen2022average}. As noted previously, all integer solutions $\mathbf{c}$ satisfying $\mathbf{c} \cdot \mathbf{x} = 0$ form a lattice $\mathcal{L}_{\mathbf{x}}^\perp$. Thus, finding a nonzero vector $\mathbf{c} \in C^n$ with $\mathbf{c} \cdot \mathbf{x} = 0$ and $\|\mathbf{c}\|_\infty \le d$ reduces to finding a short vector in the $\ell_\infty$ norm, i.e., an instance of SVP in the $\ell_\infty$ norm.

We first give a sufficient condition under which a solution exists.

\begin{theorem}
\label{thm:solexist}
Let $C=[-d:d]$ and $\mathbf{x}=(x_1,\ldots,x_n)$. If 
\[
\frac{\|\mathbf{x}\|_2}{\gcd(x_1,\ldots,x_n)} < (d+1)^{n-1},
\]
then there exists $\mathbf{c} \in C^n \setminus \{\mathbf{0}\}$ such that $\mathbf{c} \cdot \mathbf{x} = 0$.
\end{theorem}

\begin{proof}
As discussed above, it suffices to find a nonzero vector in $\mathcal{L}_{\mathbf{x}}^{\perp}$ with $\ell_\infty$ norm at most $d$. By Minkowski's theorem (\cref{lemma:minkowski}), we have
\[
\lambda^{(\infty)}_1(\mathcal{L}_{\mathbf{x}}^{\perp}) \le \det(\mathcal{L}_{\mathbf{x}}^{\perp})^{1/(n-1)}.
\]
Thus, if
\[
\det(\mathcal{L}_{\mathbf{x}}^{\perp}) <(d+1)^{n-1},
\]
then
\[
\lambda^{(\infty)}_1(\mathcal{L}_{\mathbf{x}}^{\perp}) < d+1 \Longrightarrow \lambda^{(\infty)}_1(\mathcal{L}_{\mathbf{x}}^{\perp}) \le d.
\]
This implies the existence of a nonzero vector $\mathbf{c} \in \mathcal{L}_{\mathbf{x}}^{\perp}$ with $\|\mathbf{c}\|_\infty \le d$, which yields a valid solution.
\end{proof}

\begin{comment}
\begin{proof}
Consider Minkowski's Theorem for $\mathcal{L}$ we defined above in $\ell_{\infty}$ norm, there exist a nonzero vector $\mathbf{v}=(v_0,v_1,\dots,v_n)$ satisfying
\[
\|\mathbf{v}\|_{\infty} \le \det(\mathcal{L})^{1/(n+1)}=(\alpha q)^{1/(n+1)}<d+1.
\]
Since $\|\mathbf{v}\|_{\infty}$ must be an integer, so $\|\mathbf{v}\|_{\infty}\le d$. 
We can also write
\[
\mathbf{v} = \big(\alpha (b_0q+\sum_{i=1}^n b_i x_i), b_1, \ldots, b_n \big)
\]
for some integers $b_i$. Thus we know that $\alpha|v_0$, combining with $|v_0|\le d<\alpha$ we deduce $v_0$ must be $0$. Then we have
\[
q|\sum_{i=1}^n b_ix_i,
\]
however
\[
|\sum_{i=1}^n b_ix_i|\le d\sum_{i=1}^n |x_i|<d\sum_{i=1}^n |x_i|+1=q.
\]
Finally we show $\sum_{i=1}^n b_ix_i=0$. Since $\mathbf{v}\ne \mathbf{0}$ and $v_0=0$, it follows that $(b_1,\dots,b_n)\ne \mathbf{0}$ and this is the vector we want.
\end{proof}
\end{comment}

Thus, we obtain a deterministic algorithm for subset balancing problems by employing a deterministic SVP solver. Dadush et al.~\cite{dadush2011enumerative} proposed a deterministic algorithm for SVP in any norm for full-rank lattices. To leverage this algorithm for analyzing the concrete time complexity of subset balancing problems, we propose in this section an explicit construction motivated by prior cryptanalysis literature \cite{lagarias1985solving}.

With integer parameters $\alpha,q$ , we construct the lattice $\mathcal{L}_{\alpha,q}$ with basis $\mathbf{B}_{\alpha,q}$ as  
\[
\mathbf{B}_{\alpha,q}=
\begin{pmatrix}
\alpha q \\
\alpha x_1 & 1 &        &        \\
\alpha x_2 &   & 1      &        \\
\vdots     &   &        & \ddots \\
\alpha x_n &   &        &        & 1
\end{pmatrix}
\in \mathbb{Z}^{(n+1)\times(n+1)}.
\]

Then we show the connection between the solution of Subset Balancing Problem and the short lattice vector in $\mathcal{L}_{\alpha,q}$ by the following lemma.
\begin{lemma}
\label{lem:connsolandshortvec}
If $\alpha>d, q>d\sum_{i}^n|x_i|$, then
\[
\mathcal{L}_{\alpha,q}\cap d~\mathbf{B}_{\infty}^{n+1}=\left\{(0,c_1,\dots,c_n)\Bigg|\sum_{i=1}^nc_ix_i=0,|c_i|\le d \right\}.
\]
\end{lemma}

\begin{proof}
Every $\mathbf{v}\in \mathcal{L}_{\alpha,q}$ can be written uniquely as
\[
\mathbf{v}
= \bigl(\alpha(qz_0+\sum_{i=1}^n z_i x_i),\, z_1,\dots,z_n\bigr)
\]
for some integers $z_0,z_1,\dots,z_n\in\mathbb{Z}$.

\medskip
\noindent\emph{($\subseteq$).}
Take $\mathbf{v}\in \mathcal{L}_{\alpha,q}\cap d\,\mathbf{B}_\infty^{n+1}$.
Write $\mathbf{v}=(v_0,v_1,\dots,v_n)$ as above, so $v_i=z_i$ for $i=1,\dots,n$ and
\[
v_0=\alpha\Bigl(qz_0+\sum_{i=1}^n z_i x_i\Bigr).
\]
Since $\|\mathbf{v}\|_\infty\le d$, we have $|z_i|=|v_i|\le d$ for all $i=1,\dots,n$.
Moreover, because $\alpha>d$, the inequality $|v_0|\le d$ forces $v_0=0$, hence
\[
qz_0+\sum_{i=1}^n z_i x_i=0.
\]
If $z_0\neq 0$, then
\[
\Bigl|\sum_{i=1}^n z_i x_i\Bigr| = q|z_0|\ge q.
\]
On the other hand, by $|z_i|\le d$,
\[
\Bigl|\sum_{i=1}^n z_i x_i\Bigr|
\le \sum_{i=1}^n |z_i||x_i|
\le d\sum_{i=1}^n |x_i|.
\]
These two inequalities contradict the assumption $q>d\sum_{i=1}^n|x_i|$.
Therefore $z_0=0$, and consequently $\sum_{i=1}^n z_i x_i=0$.
Thus $\mathbf{v}=(0,z_1,\dots,z_n)$ with $\sum_{i=1}^n z_i x_i=0$ and $|z_i|\le d$.

\medskip
\noindent\emph{($\supseteq$).}
Conversely, let $(0,c_1,\dots,c_n)$ be any integer vector with $\sum_{i=1}^n c_i x_i=0$
and $|c_i|\le d$.
Set $z_0=0$ and $z_i=c_i$ for $i=1,\dots,n$. Then the lattice parametrization gives
\[
(0,c_1,\dots,c_n)
=\bigl(\alpha(\sum_{i=1}^n c_i x_i),\,c_1,\dots,c_n\bigr)\in \mathcal{L}_{\alpha,q}.
\]
Since $|c_i|\le d$ and the first coordinate is $0$, this vector lies in
$d\,\mathbf{B}_\infty^{n+1}$.

\medskip
Combining both directions proves the claimed equality.   
\end{proof}

Then we show a dimension preserving reduction from Subset Balancing Problem with $C=[-d:d]$ to $\text{SVP}_{\infty}$.

\begin{theorem}
\label{thm:redtosvp}
Given any $d,n>0$, Subset Balancing on $C=[-d:d]$ can be solved by a single call of a $\text{SVP}_{\infty}$ oracle in dimension $n+1$.
\end{theorem}

\begin{proof}

Set $\alpha= d+1, q = d\sum_{i=1}^n |x_i|+1$, so \(\alpha>d\) and \(q>d\sum_i |x_i|\). Construct \(\mathcal{L}_{\alpha,q}\subseteq \mathbb{Z}^{n+1}\)
with basis \(\mathbf{B}_{\alpha,q}\), and call the \(\mathrm{SVP}_\infty\) oracle on \(\mathcal{L}_{\alpha,q}\).
Let \(\mathbf{v}=(v_0,v_1,\dots,v_n)\) be the returned shortest nonzero vector.

Output \((c_1,\dots,c_n)=(v_1,\dots,v_n)\) iff \(v_0=0\) and \(\|\mathbf{v}\|_\infty\le d\);
otherwise report that no solution exists.

\paragraph{Correctness and Soundness.}
If there exists \((c_1,\dots,c_n)\ne\mathbf{0}\) with \(\sum_i c_i x_i=0\) and \(|c_i|\le d\), then by
Lemma~\ref{lem:connsolandshortvec},
\((0,c_1,\dots,c_n)\in \mathcal{L}_{\alpha,q}\cap d\,\mathbf{B}_\infty^{n+1}\).
Hence \(\lambda_1^{(\infty)}(\mathcal{L}_{\alpha,q})\le d\), so the oracle returns \(\mathbf{v}\) with
\(\|\mathbf{v}\|_\infty\le d\), and again Lemma~\ref{lem:connsolandshortvec} implies \(v_0=0\);
thus the algorithm outputs a valid solution. If no such solution exists, then the intersection
contains no nonzero vector, so \(\lambda_1^{(\infty)}(\mathcal{L}_{\alpha,q})>d\) and the algorithm reports that no solution exists.

The reduction makes one \(\mathrm{SVP}_\infty\) call in dimension \(n+1\).
\end{proof}

\paragraph{Proof of \cref{thm:mainworstcase}.}

    Combining Theorems \ref{thm:redtosvp} and \ref{thm:svp-infty}, there is a deterministic algorithm that solves Subset Balancing with $C=[-d:d]$ in time
$\tilde{O}\!\left(\bigl(6\sqrt{2\pi e}\,\bigr)^n\right)\approx \tilde{O}(2^{4.632n})$
for any $d>0$.
Moreover, combining Theorem \ref{thm:redtosvp} with Mukhopadhyay's $\tilde{O}(2^{2.443n})$ randomized algorithm for $\text{SVP}_{\infty}$~\cite{mukhopadhyay2019faster}, we obtain a randomized algorithm that solves Subset Balancing with $C=[-d:d]$ in time $\tilde{O}(2^{2.443n})$.\qed

\begin{remark}
Notably, the time complexity of our algorithm is independent with the size of $C$ (ignoring polynomial factor). The randomized algorithm with time complexity $$\tilde{O}(2^{2.443n})$$ outperforms $O(|C|^{(0.5-\varepsilon)n})$ in~\cite{randolph2025beating} when $d > 14$, and its performance could improve further if more efficient SVP algorithms in the $\ell_\infty$ norm become available.

Moreover, if we employ the fastest heuristic $\text{SVP}_{\infty}$ algorithm analyzed in \cite{aggarwal2018improved} of time complexity $\tilde{O}(2^{0.62n})\approx \tilde{O}(1.537^{n})$, which could beat the fastest $\tilde{O}(1.7067^{n})$ in~\cite{randolph2025beating} even when $d=1$.

\end{remark}

\paragraph{Polynomial-time algorithms for sufficiently large $d$.}

Our motivation is that when $d$ is sufficiently large, one can replace the SVP oracle with the LLL algorithm~\cite{lenstra1982factoring} to obtain a relatively short vector, which suffices to produce a valid solution for the Subset Balancing problem.

\begin{theorem}
\label{thm:sbppoly}
Let $C=[-d:d]$ and $\mathbf{x}=(x_1,\ldots,x_n)\in\mathbb{Z}^n$. Suppose 
$$
d > 2^{\frac{n-2}{4}} \cdot \left(\frac{\|\mathbf{x}\|_2}{\gcd(x_1,\ldots,x_n)}\right)^{1/(n-1)}-1.
$$
then we could find $\mathbf{c} \in C^n \setminus \{\mathbf{0}\}$ such that $\mathbf{c} \cdot \mathbf{x} = 0$ in time polynomial in $n$ and the bit length of $\mathbf{x}$.
\end{theorem}

\begin{proof}
We run LLL algorithm with the lattice $\mathcal{L}^{\perp}_{\mathbf{x}}$. By~\cref{lemma:LLL}, we obtain a nonzero vector $\mathbf{v} \in \mathcal{L}^{\perp}_{\mathbf{x}}$ satisfying
\[
\|\mathbf{v}\|_2 \le 2^{\frac{n-2}{4}} \det(\mathcal{L}^{\perp}_{\mathbf{x}})^{\frac{1}{n-1}}<d+1.
\] 
Since $\|\mathbf{v}\|_\infty \le \|\mathbf{v}\|_2<d+1$, and $\|\mathbf{v}\|_\infty$ is a integer, thus $\|\mathbf{v}\|\le d$ is a valid solution.

The running time is dominated by the running time of LLL algorithm, which is polynomial in $n$ and the bit length of $\mathbf{x}$.
\end{proof}

\subsection{Generalization of Coefficient Set}
The original \emph{Subset Balancing Problem} can be equivalently stated as follows: given the convex set $[-d,d]^n$, find a nonzero integer vector $\mathbf{c}$ such that
\[
\mathbf{c} \in \bigl(\mathbb{Z}^n \cap [-d,d]^n\bigr)\setminus \{\mathbf{0}\}
\quad\text{and}\quad
\mathbf{c}\cdot \mathbf{x} = 0.
\]
In this subsection, we generalize the specific box $[-d,d]^n$ to an arbitrary convex and centrally symmetric set $K$ with $\mathbf{0}\in K$, and we give a deterministic algorithm running in single-exponential in $n$, i.e., $\tilde{O}(2^{c_K n})$ time. Here the constant $c_K$ is independent of the size of $K$, and depends only on the shape of $K$.

We may consider the Minkowski functional induced by $K$:
\[
\|\mathbf{x}\|_{K} \;:=\; \inf\{\, r\ge 0 : \mathbf{x}\in rK \,\}.
\]
In our settings that $K$ is convex, centrally symmetric, and contains the origin in its interior, thus $\|\cdot\|_{K}$ is a norm. 

\begin{theorem}
Given a convex and centrally symmetric set $K\in \mathbb{R}^n$ with $\mathbf{0}\in K$, and $\mathbf{x}\in \mathbb{Z}^n$, there is a deterministic algorithm that either finds a vector $\mathbf{c}$ such that
\[
\mathbf{c} \in \bigl(\mathbb{Z}^n \cap K\bigr)\setminus \{\mathbf{0}\}
\quad\text{and}\quad
\mathbf{c}\cdot \mathbf{x} = 0,
\]
or reports that no such solution exists. The running time is
\[
\tilde{O}(2^{c_K n}),
\]
where the constant $c_K$ is independent of the size of $K$, and depends only on the shape of $K$.
\end{theorem}

\begin{proof}
Recall that for $\mathbf{x}=(x_1,\ldots,x_n)$, the set of integer vectors 
$\mathbf{c}=(c_1,\ldots,c_n)$ satisfying
\[
\sum_{i=1}^n c_i x_i = 0
\]
forms a rank $n-1$ sublattice $\mathcal{L}^{\perp}_{\mathbf{x}}$. 

However, Theorem~\ref{thm:svpanynorm} requires a full rank lattice that we could not directly apply it on $\mathcal{L}^{\perp}_{\mathbf{x}}$. We use a small trick to address this. Take the minimal $d\in\mathbb{Z}$ such that $K\subseteq [-d,d]^n$.

Set $q=d\sum_{i=1}^n |x_i|+1$, and take an arbitrary $i_0$ such that $x_{i_0}\ne 0$. Then consider the full rank lattice
\[
\mathcal{L}=\mathcal{L}^{\perp}_{\mathbf{x}}+\mathbb{Z} (q\mathbf{e}_{i_0}),
\]
where $\mathbf{e}_i$ is $i_0$-th unit vector.

We claim that $\mathcal{L}\cap K=\mathcal{L}^{\perp}_{\mathbf{x}}\cap K$. Otherwise, there exists $\mathbf{v}_0\in (\mathcal{L}\setminus\mathcal{L}^{\perp}_{\mathbf{x}})\cap K$. We could also write 
\[
\mathbf{v}_0=bq\mathbf{e}_i+\mathbf{c}, 
\]
with $b\neq 0,c\in\mathcal{L}^{\perp}_{\mathbf{x}}$. Then it follows
\[
\langle \mathbf{v}_0, \mathbf{x} \rangle=\langle bq\mathbf{e}_i, \mathbf{x} \rangle=bqx_i.
\]
However, since $\mathbf{v}_0\in K \subseteq [-d,d]^n$, 
\[
|\langle \mathbf{v}_0, \mathbf{x} \rangle|\le d\sum_{i=1}^n|x_i|<q,
\]
contradicting with $bx_i\ne 0$.

Our algorithm simply calls an SVP oracle (Theorem~\ref{thm:svpanynorm}) of $\mathcal{L}$ under the norm $\|{\cdot}\|_K$. Suppose the SVP oracle returns $\mathbf{v}$, then checks whether $\|\mathbf{v}\|\in K$, if yes, returns $\mathbf{v}$; otherwise, reports that no solution exists.

\paragraph{Correctness and Soundness.}
If there exists a solution $\mathbf{c}$, then $\|\mathbf{v}\|_{K}\le \|\mathbf{c}\|_{K}\le 1$, which means $\mathbf{v}\in K$. As $\mathcal{L}\cap K=\mathcal{L}^{\perp}_{\mathbf{x}}\cap K$, we conclude $\mathbf{v}\in \mathcal{L}^{\perp}_{\mathbf{x}}$ as well, thus yields a valid solution.

On the other case, there is no solution, so the SVP oracle must return a vector out of $K$, and the algorithm reports no solution exists.

The time is dominated by the running time of SVP oracle of norm $\|\cdot\|_{K}$, which is
\[
\tilde{O}(2^{c_K n}),
\]
for some constant $c_K$.
\end{proof}

\section{Worst Case Algorithms for Generalized Subset Sum}
\label{sec:worstcasegss}

At SODA 2022, Chen, Jin, Randolph, and Servedio~\cite{chen2022average} proposed a generalized version of the Subset Sum problem. Given an input bound $M$, a vector $\mathbf{x} = (x_1, \dots, x_n) \in [0, M-1]^n$, a set $C \subset \mathbb{Z}$ of allowed coefficients, and a target integer $\tau$, the goal is to find a coefficient vector $\mathbf{c} \in C^n$ such that $\mathbf{c} \cdot \mathbf{x} = \tau$, if such a vector exists.

In \cref{subsecgss:-d:d}, we first show that the worst case of Generalized Subset Sum with $C=[-d:d]$ could reduce to $\text{CVP}_{\infty}$. However, as far as we know, there is no single exponential time algorithm for $\text{CVP}_{\infty}$. 

We observe that the distance from the target to any lattice point can take only integer values, which lets us replace an exact $\mathrm{CVP}_\infty$ oracle by an $\left(1+\frac{1}{d}-\varepsilon\right)$ approximate one for the $C=[-d:d]$ case. Then we apply the oracle proposed by Eisenbrand, H\"ahnle, and Niemeier~\cite{eisenbrand2011covering} in time $2^{O(n)} \cdot (\log(1/\varepsilon))^{O(n)}$, yielding the overall complexity $2^{O(n \log\log d)}$.

\subsection{Worst Case Generalized Subset Sum with $C=[-d:d]$}
\label{subsecgss:-d:d}
We construct the same structure lattice $\mathcal{L}$ as previous section, more precisely, we denote the lattice $\mathcal{L}_{\alpha,q}$ with basis $\mathbf{B}_{\alpha,q}$ as  
\[
\mathbf{B}_{\alpha,q}=
\begin{pmatrix}
\alpha q \\
\alpha x_1 & 1 &        &        \\
\alpha x_2 &   & 1      &        \\
\vdots     &   &        & \ddots \\
\alpha x_n &   &        &        & 1
\end{pmatrix}
\in \mathbb{Z}^{(n+1)\times(n+1)}.
\]
Define
\[
\mathrm{dist}_{\infty}(\mathbf{t}, \mathcal{L})
=\min_{\mathbf{v}\in \mathcal{L}}\|\mathbf{v}-\mathbf{t}\|_{\infty}.
\]

\begin{lemma}
\label{lem:connsolclosevec}
Set the target vector $\mathbf{t}=(\alpha\tau,0,\dots,0)$, for any integer $\alpha>d,q>d\sum_{i=1}^n|x_i|+|\tau|$,
\[
\mathcal{L}_{\alpha,q}\cap (\mathbf{t}+d~\mathbf{B}_{\infty}^{n+1})=\left\{(\alpha\tau,c_1,\dots,c_n)\Bigg|\sum_{i=1}^nc_ix_i=\tau,|c_i|\le d\right\}.
\]
\end{lemma}

\begin{proof}
This proof is similar as the proof of lemma~\ref{lem:connsolandshortvec}. Every vector $\mathbf{v}\in \mathcal{L}_{\alpha,q}$ can be written as
\[
\mathbf{v}
= \bigl(\alpha(qz_0+\sum_{i=1}^n z_i x_i),\, z_1,\dots,z_n\bigr)
\]
for some integers $z_0,z_1,\dots,z_n\in\mathbb{Z}$.

\medskip
\noindent\emph{($\subseteq$).}
Let $\mathbf{v}\in \mathcal{L}_{\alpha,q}\cap (\mathbf{t}+d\,\mathbf{B}_\infty^{n+1})$.
Then $\|\mathbf{v}-\mathbf{t}\|_\infty\le d$, i.e.,
\[
|v_0-\alpha\tau|\le d
\quad\text{and}\quad
|v_i|\le d \ \text{for all } i=1,\dots,n.
\]
From the lattice parametrization we have $v_i=z_i$ for $i=1,\dots,n$, hence
$|z_i|\le d$ for all $i=1,\dots,n$, and
\[
\bigl|v_0-\alpha\tau\bigr|
= \alpha\Bigl|\Bigl(qz_0+\sum_{i=1}^n z_i x_i\Bigr)-\tau\Bigr|
\le d.
\]
Since $\alpha>d$, the inequality above implies
\[
\Bigl|\Bigl(qz_0+\sum_{i=1}^n z_i x_i\Bigr)-\tau\Bigr|=0 \Longrightarrow qz_0+\sum_{i=1}^n z_i x_i=\tau.
\]
If $z_0\neq 0$, then
\[
\Bigl|\tau-\sum_{i=1}^n z_i x_i\Bigr|=|qz_0|\ge q.
\]
On the other hand, using $|z_i|\le d$ we have
\[
\Bigl|\tau-\sum_{i=1}^n z_i x_i\Bigr|
\le |\tau|+\sum_{i=1}^n |z_i||x_i|
\le |\tau|+d\sum_{i=1}^n |x_i|.
\]
This contradicts the assumption \(q>d\sum_{i=1}^n|x_i|+|\tau|\). Hence $z_0=0$, and thus $\sum_{i=1}^n z_i x_i=\tau$. So $\mathbf{v}=(\alpha\tau,z_1,\dots,z_n)$ with $|z_i|\le d$.

\medskip
\noindent\emph{($\supseteq$).}
Conversely, let $(\alpha\tau,c_1,\dots,c_n)$ be such that $\sum_{i=1}^n c_i x_i=\tau$
and $|c_i|\le d$.
Set $z_0=0$ and $z_i=c_i$ for $i=1,\dots,n$.
Then
\[
(\alpha\tau,c_1,\dots,c_n)
=\bigl(\alpha(\sum_{i=1}^n c_i x_i),\,c_1,\dots,c_n\bigr)\in \mathcal{L}_{\alpha,q}.
\]
Moreover,
\(\|(\alpha\tau,c_1,\dots,c_n)-(\alpha\tau,0,\dots,0)\|_\infty\le d\),
so the vector lies in \(\mathbf{t}+d\,\mathbf{B}_\infty^{n+1}\).

\medskip
Combining the two directions yields the claim.
\end{proof}

It is easy to show that the worst case of Generalized Subset Sum with $C=[-d:d]$ could reduce to $\text{CVP}_{\infty}$. However, as far as we know, there is no single exponential time algorithm for $\text{CVP}_{\infty}$. 

To overcome this for the worst case GSS, our key observation is that the distances are discrete integers. We can exploit this gap: if a solution exists within distance $d$, then we conclude $\mathrm{dist}_{\infty}(\mathbf{t}, \mathcal{L})\le d$; otherwise no solution exists, then $\mathrm{dist}_{\infty}(\mathbf{t}, \mathcal{L})\ge d+1$. This means we only need an approximate $\text{CVP}_{\infty}$ oracle with approximation factor $\left(1+\frac{1}{d}-\varepsilon\right)$ for arbitrarily small $\varepsilon>0$.

\begin{theorem}
\label{thm:redtoapproxcvp}
Given any $d,n>0$, Generalized Subset Sum on $C=[-d:d]$ can be solved by a single call of a $(1+1/d-\varepsilon)$-approximate $\text{CVP}_{\infty}$ oracle in dimension $n+1$ for any $0<\varepsilon<1/d$.
\end{theorem}
\begin{proof}
Choose parameters $\alpha = d+1,q = d\sum_{i=1}^n |x_i| + |\tau| + 1$, so that \(\alpha>d\) and \(q>d\sum_{i=1}^n |x_i|+|\tau|\).
Construct the lattice \(\mathcal{L}_{\alpha,q}\) with basis \(\mathbf{B}_{\alpha,q}\), and set the target
\[
\mathbf{t}= (\alpha\tau,0,\dots,0)\in\mathbb{Z}^{n+1}.
\]
Make a single call to a $(1+1/d-\varepsilon)$-approximate \(\mathrm{CVP}_\infty\) oracle on input \((\mathcal{L}_{\alpha,q},\mathbf{t})\),
and let \(\mathbf{v}=(v_0,v_1,\dots,v_n)\in\mathcal{L}_{\alpha,q}\) be the lattice vector returned by the oracle.

The reduction outputs \((c_1,\dots,c_n)=(v_1,\dots,v_n)\) if $\|\mathbf{v}-\mathbf{t}\|_\infty \le d$, otherwise it reports no solution exists.

\paragraph{Correctness and Soundness.}

Assume there exists \(\mathbf{c}=(c_1,\dots,c_n)\in C^n\) such that \(\sum_{i=1}^n c_i x_i=\tau\).
Then Lemma~\ref{lem:connsolclosevec} gives
\[
(\alpha\tau,c_1,\dots,c_n)\in \mathcal{L}_{\alpha,q}\cap(\mathbf{t}+d\,\mathbf{B}_\infty^{n+1}),
\]
so \(\mathrm{dist}_\infty(\mathbf{t},\mathcal{L}_{\alpha,q})\le d\).
Therefore the $1+1/d-\varepsilon$-approximate \(\mathrm{CVP}_\infty\) oracle must return some \(\mathbf{v}\in\mathcal{L}_{\alpha,q}\) with
\[
\|\mathbf{v}-\mathbf{t}\|_\infty\le\left(1+\frac{1}{d}-\varepsilon\right)\mathrm{dist}_\infty(\mathbf{t},\mathcal{L}_{\alpha,q})< d+1
\Longrightarrow \|\mathbf{v}-\mathbf{t}\|_\infty\le d,
\]
and the algorithm will accept.
Moreover, by Lemma~\ref{lem:connsolclosevec}, this accepted \(\mathbf{v}\) corresponds to coefficients
\((v_1,\dots,v_n)\) satisfying \(\sum_i v_i x_i=\tau\) and \(|v_i|\le d\), so the output is valid.

If no such \(\mathbf{c}\) exists, then Lemma~\ref{lem:connsolclosevec} implies
\(\mathcal{L}_{\alpha,q}\cap(\mathbf{t}+d\,\mathbf{B}_\infty^{n+1})=\emptyset\), i.e.
\(\mathrm{dist}_\infty(\mathbf{t},\mathcal{L}_{\alpha,q})\ge d+1\).
Thus even the closest vector \(\mathbf{v}\) must satisfy \(\|\mathbf{v}-\mathbf{t}\|_\infty>d\), and the algorithm reports no solution exists.

This uses a $(1+1/d-\varepsilon)$-approximate \(\mathrm{CVP}_\infty\) oracle in dimension \(n+1\).
\end{proof}

Eisenbrand, H\"ahnle, and Niemeier~\cite{eisenbrand2011covering} gave a randomized algorithm for $(1+\varepsilon)$-approximate CVP$_\infty$ that runs in time 
$$
2^{O(n)} \cdot (\log(1/\varepsilon))^{O(n)},
$$
where $n$ is the lattice dimension. Their algorithm works by 
constructing a $(2, \varepsilon)$-covering of the $\ell_\infty$ 
ball using $O((1 + \log(1/\varepsilon))^n)$ parallelepipeds, 
and invoking a $2$-approximate CVP$_\infty$ oracle for each.

\smallskip\noindent\textbf{Derandomization.}
The covering construction in~\cite{eisenbrand2011covering} is deterministic 
based on a coordinate-wise binary partition of the cube. Their original algorithm invokes a randomized $2$-approximate CVP$_\infty$ oracle by~\cite{blomer2009sampling} as a subroutine. We instantiate this oracle with the deterministic $2^{O(n)}$-time algorithm of Dadush and Kun~\cite{dadush2013lattice} by setting $\varepsilon=1$ in their algorithm. 
This yields a fully deterministic algorithm.

\begin{proposition}[Deterministic Approximate CVP~\cite{eisenbrand2011covering}]
\label{pro:approxcvp}
There is a deterministic algorithm for $(1+\varepsilon)$-approximate CVP in $\ell_{\infty}$ norm in time
\[
2^{O(n)} (\log(1/\varepsilon))^{O(n)}.
\]
\end{proposition}
The proof can be found at appendix~\cref{sec:appendixapproxcvp}.

\begin{theorem}[Algorithm for Worst-Case GSS with \texorpdfstring{$C=[-d:d]$}{C=[-d:d]}]
\label{thm:worstcasegss-d:d}
Given any $d>2,n \in \mathbb{N}$ and let $C = [-d:d]$, there is a deterministic algorithm for Generalized Subset Sum in running time
\[
2^{O(n\log\log d)}.
\]    
\end{theorem}

\begin{proof}
Combining Theorem~\ref{thm:redtoapproxcvp} and Proposition~\ref{pro:approxcvp}, for $d>2,\log\log d>0$, set the approximate factor $1+1/(2d)$ completes the proof.
\end{proof}

\subsection{Worst Case Generalized Subset Sum with $C=[\pm d]$}
\label{subsecgss:pmd}
In the case when $C=[\pm d]$, it seems much harder for lattice algorithm. In previous seting $C=[-d:d]$, the solution has a one-to-one correspondence with the shortest (or closest) vector. However, in the situation that $C=[\pm d]$, the closest vector may corresponding to an invalid solution.

\paragraph{Warm-up: from \(C=[-d:d]\) to a general interval \(C=[a:b]\).}
Before proving the reduction formally, we note that essentially nothing changes when we replace the symmetric range \([-d:d]\) by an arbitrary integer interval \([a:b]\): we keep the same lattice \(\mathcal{L}_{\alpha,q}\) and only shift the target to recenter the allowed coefficient window.
Concretely, letting \(\mu=(a+b)/2\) and \(r=(b-a)/2\), we replace \(\mathbf{t}=(\alpha\tau,0,\ldots,0)\) by
\[
\mathbf{t}=(\alpha\tau,\mu,\ldots,\mu),
\]
so that the constraint \(c_i\in[a:b]\) is captured by the cube condition \(\|\,(c_1,\ldots,c_n)-(\mu,\ldots,\mu)\,\|_\infty\le r\).
With appropriate choices of \(\alpha,q\), this yields 
\[
\mathcal{L}_{\alpha,q}\cap (\mathbf{t}+r~\mathbf{B}_{\infty}^{n+1})=\left\{(\alpha\tau,c_1,\dots,c_n)\Bigg|\sum_{i=1}^nc_ix_i=\tau,a\le c_i\le b\right\},
\]
and therefore feasibility reduces to a single \(\mathrm{CVP}_\infty\) instance. Moreover, since the \(\ell_\infty\)-distance is discrete, we get a gap:
\[
\mathrm{dist}_\infty(\mathbf{t},\mathcal{L}_{\alpha,q})\le r \quad\text{(YES-instance)},\qquad
\mathrm{dist}_\infty(\mathbf{t},\mathcal{L}_{\alpha,q})\ge r+1 \quad\text{(NO-instance)}.
\]
Hence a \(\bigl(\tfrac{r+1}{r}-\varepsilon\bigr)=\bigl(\tfrac{b-a+2}{b-a}-\varepsilon\bigr)\)-approximate \(\mathrm{CVP}_\infty\) oracle suffices.

\paragraph{From one interval to \(C=[\pm d]\) via \(2^n\) shifted targets.}
When \(C=[\pm d]=[-d,-1]\cup[1,d]\), we view \(C^n\) as a union of boxes aligned with the \(2^n\) axis, one for each choice of coordinate sign.
For each sign pattern \(\mathbf{s}\in\{\pm1\}^n\), we apply the interval routine to the corresponding box by using a target whose last \(n\) coordinates are shifted to the appropriate midpoint:
\[
\mathbf{t}_{\mathbf{s}} = \bigl(\alpha\tau,\, s_1\cdot\tfrac{d+1}{2},\,\ldots,\, s_n\cdot\tfrac{d+1}{2}\bigr),
\qquad r=\tfrac{d-1}{2},
\]
while keeping \(\mathcal{L}_{\alpha,q}\) unchanged.
Thus worst-case GSS for \(C=[\pm d]\) is solved by enumerating \(\mathbf{s}\) and making \(2^n\) calls to an approximate \(\mathrm{CVP}_\infty\) oracle; the same distance-gap argument shows that approximation factor
\[
\frac{r+1}{r}-\varepsilon \;=\; \frac{d+1}{d-1}-\varepsilon
\]
is enough for each call. Plugging in the known deterministic \((1+\varepsilon)\)-approximate \(\mathrm{CVP}_\infty\) algorithm (Proposition~\ref{pro:approxcvp}) then gives overall time \(2^{O(n\log\log d)}\) (with the extra \(2^n\) factor absorbed into the exponent).

\begin{theorem}[Algorithm for Worst-Case GSS with \texorpdfstring{$C=[\pm d]$}{C=[\pm d]}]
\label{thm:worstcasegsspmd}
Given any $d>3,n \in \mathbb{N}$ and let $C = [\pm d]$, there is a deterministic algorithm for Generalized Subset Sum in running time
\[
2^{O(n\log\log d)}.
\]    
\end{theorem}
\begin{proof}
Let \(r = (d-1)/2\). Choose the same parameters $\alpha = d+1,q = d\sum_{i=1}^n |x_i| + |\tau| + 1$, so that \(\alpha>d\) and \(q>d\sum_{i=1}^n |x_i|+|\tau|\). Construct the lattice \(\mathcal{L}_{\alpha,q}\) with basis \(\mathbf{B}_{\alpha,q}\) as in \cref{subsecgss:-d:d}.

\paragraph{Algorithm.}
For each sign pattern \(\mathbf{s}=(s_1,\ldots,s_n)\in\{\pm 1\}^n\), define the target
\[
\mathbf{t}_{\mathbf{s}}
= \bigl(\alpha\tau,\; s_1\cdot\tfrac{d+1}{2},\;\ldots,\; s_n\cdot\tfrac{d+1}{2}\bigr)\in \mathbb{R}^{n+1}.
\]
Invoke a \(\gamma\)-approximate \(\mathrm{CVP}_\infty\) oracle on \((\mathcal{L}_{\alpha,q},\mathbf{t}_{\mathbf{s}})\), where we will set
\[
\gamma = 1+\frac{1}{d-1} = 1+\frac{1}{2r}.
\]
Let \(\mathbf{v}=(v_0,\ldots,v_n)\in\mathcal{L}_{\alpha,q}\) be the returned lattice vector.
If \(\|\mathbf{v}-\mathbf{t}_{\mathbf{s}}\|_\infty \le r\), output \((c_1,\ldots,c_n):=(v_1,\ldots,v_n)\) and halt.
If no \(\mathbf{s}\) yields acceptance, reports that no solution exists.

\paragraph{A sign-pattern version of Lemma~\ref{lem:connsolclosevec}.}
Fix \(\mathbf{s}\in\{\pm 1\}^n\). The proof of Lemma~\ref{lem:connsolclosevec} extends verbatim after shifting the last \(n\) target coordinates from \(0\) to \(s_i\cdot\tfrac{d+1}{2}\):
one shows that
\[
\mathcal{L}_{\alpha,q}\cap\bigl(\mathbf{t}_{\mathbf{s}} + r\,\mathbf{B}_\infty^{n+1}\bigr)
=
\Bigl\{(\alpha\tau,c_1,\ldots,c_n)\;\Bigm|\;
\sum_{i=1}^n c_i x_i=\tau,\ \ s_i c_i\in[1:d]\Bigr\}.
\]
In particular, \(\mathcal{L}_{\alpha,q}\cap(\mathbf{t}_{\mathbf{s}} + r\,\mathbf{B}_\infty^{n+1})\neq\emptyset\)
if and only if there exists a feasible GSS solution whose coordinate-wise signs are given by \(\mathbf{s}\).

\paragraph{Distance gap.}
Fix \(\mathbf{s}\in\{\pm1\}^n\) and write \(\mathbf{t}_{\mathbf{s}}=(\alpha\tau,\, s_1\frac{d+1}{2},\ldots,s_n\frac{d+1}{2})\).
For any \(\mathbf{v}\in\mathcal{L}_{\alpha,q}\), the first-coordinate difference satisfies
\[
v_0-\alpha\tau \in \alpha\mathbb{Z},
\]
since \(v_0=\alpha(qz_0+\sum_i z_i x_i)\).
With our choice \(\alpha=d+1\) and \(r=(d-1)/2\), we have \(r+1=(d+1)/2<\alpha\). Hence
\[
\|\mathbf{v}-\mathbf{t}_{\mathbf{s}}\|_\infty<r+1
\ \Longrightarrow\ 
|v_0-\alpha\tau|<\alpha
\ \Longrightarrow\ 
v_0=\alpha\tau,
\]
so the \(\ell_\infty\)-distance is determined entirely by the last \(n\) coordinates:
\[
\|\mathbf{v}-\mathbf{t}_{\mathbf{s}}\|_\infty
=\max_{i\in[n]}\Bigl|v_i-s_i\frac{d+1}{2}\Bigr|.
\]
Now \(v_i\in\mathbb{Z}\), while \(s_i\frac{d+1}{2}\) is an integer if \(d\) is odd and a half-integer if \(d\) is even. Therefore the above maximum ranges over \(\mathbb{Z}\) (odd \(d\)) or \(\mathbb{Z}+\tfrac12\) (even \(d\)). Since \(r=(d-1)/2\) has the same type, there is no attainable value in \((r,r+1)\), and thus
\[
\|\mathbf{v}-\mathbf{t}_{\mathbf{s}}\|_\infty<r+1
\quad\Longrightarrow\quad
\|\mathbf{v}-\mathbf{t}_{\mathbf{s}}\|_\infty\le r.
\]

\paragraph{Correctness and Soundness.}
Suppose there exists \(\mathbf{c}\in[\pm d]^n\) with \(\sum_i c_i x_i=\tau\).
Let \(s_i:=\mathrm{sign}(c_i)\in\{\pm 1\}\).
Then \(s_i c_i\in[1:d]\) for all \(i\), and by the sign-pattern lemma above,
\((\alpha\tau,c_1,\ldots,c_n)\in \mathcal{L}_{\alpha,q}\cap(\mathbf{t}_{\mathbf{s}}+r\,\mathbf{B}_\infty^{n+1})\),
so \(\mathrm{dist}_\infty(\mathbf{t}_{\mathbf{s}},\mathcal{L}_{\alpha,q})\le r\).
By \(\gamma\)-approximate \(\mathrm{CVP}_\infty\),
\[
\|\mathbf{v}-\mathbf{t}_{\mathbf{s}}\|_\infty
\le \gamma\cdot \mathrm{dist}_\infty(\mathbf{t}_{\mathbf{s}},\mathcal{L}_{\alpha,q})
\le \Bigl(1+\frac{1}{2r}\Bigr)r
= r+\frac{1}{2}
< r+1.
\]
By the gap observation, this implies \(\|\mathbf{v}-\mathbf{t}_{\mathbf{s}}\|_\infty\le r\), so the algorithm accepts for this \(\mathbf{s}\).
Finally, acceptance implies \(\mathbf{v}\in\mathcal{L}_{\alpha,q}\cap(\mathbf{t}_{\mathbf{s}}+r\,\mathbf{B}_\infty^{n+1})\),
and by the sign-pattern lemma the output coefficients \((v_1,\ldots,v_n)\) satisfy \(\sum_i v_i x_i=\tau\) and \(v_i\in[\pm d]\).

If the instance has no solution in \([\pm d]^n\), then for every \(\mathbf{s}\) the intersection
\(\mathcal{L}_{\alpha,q}\cap(\mathbf{t}_{\mathbf{s}}+r\,\mathbf{B}_\infty^{n+1})\) is empty.
Equivalently, \(\mathrm{dist}_\infty(\mathbf{t}_{\mathbf{s}},\mathcal{L}_{\alpha,q})>r\), so no lattice vector can satisfy
\(\|\mathbf{v}-\mathbf{t}_{\mathbf{s}}\|_\infty\le r\).
Hence the algorithm never accepts and correctly reports that no solution exists.

\paragraph{Running time.}
We make \(2^n\) approximate \(\mathrm{CVP}_\infty\) oracle calls in dimension \(n+1\).
For $d>3$, instantiating the oracle by Proposition~\ref{pro:approxcvp} with \(\varepsilon=1/(d-1)\) gives per-call time
\[
2^{O(n)}\bigl(\log(1/\varepsilon)\bigr)^{O(n)}
=
2^{O(n)}\bigl(\log(d-1)\bigr)^{O(n)}
=
2^{O(n\log\log d)}.
\]
Multiplying by the outer \(2^n\) factor yields overall time \(2^{O(n\log\log d)}\), as claimed.
\end{proof}

\paragraph{Proof of \cref{thm:worstcasegss}.}
Combining \cref{thm:worstcasegss-d:d} and \cref{thm:worstcasegsspmd} completes the proof.\qed

\section{Average Case Algorithms for Generalized Subset Sum}
\label{sec:averagecasealg}

In \cref{subsecgss:-d:d}, we already show that the worst case of Generalized Subset Sum with $C=[-d:d]$ could reduce to an approximate $\text{CVP}_{\infty}$ oracle in time $2^{O(n\log\log d)}$. In this section, we consider the average-case that $\mathbf{x}$ is uniformly random over $[0:M-1]^n$ and $|\tau| = o(Mn)$ as~\cite{chen2022average}, where we show we can solve it with single exponential time without $d$ in \cref{subsec:-d:d}. And we also explore the case where $C=[\pm d]$ in \cref{subsec:pmd}.

\subsection{Average Case Generalized Subset Sum with $C=[-d:d]$}
\label{subsec:-d:d}

In this section we consider the average-case that $C=[-d:d]$, $\mathbf{x}$ is uniformly random over $[0:M-1]^n$ as settings in~\cite{chen2022average}.

Our main idea is to show that $\mathrm{dist}_\infty(\mathbf{t},\mathcal{L}_{\alpha,q})$ is bounded by a constant multiple of $\lambda_1^{(\infty)}(\mathcal{L}_{\alpha,q})$ with high probability. To this end, we use the following theorem to bound both $\lambda_1^{(\infty)}(\mathcal{L}_{\alpha,q})$ and $\mathrm{dist}_\infty(\mathbf{t},\mathcal{L}_{\alpha,q})$.

\begin{theorem}[Theorem 4 in~\cite{chen2022average}]
\label{thm:solprobC0}
Let $C=[-d:d]$ for some fixed $d\in \N$, and fix any constant $\epsilon >0$.
For $\mathbf{x} \sim [0:M-1]^n$ and any integer $\tau$  satisfying $|\tau|=o(Mn)$, we have
\[
\mathop{{\bf Pr}}_{\mathbf{x}}\Big[\exists\hspace{0.05cm} \mathbf{c} \in C^n\setminus\{\mathbf{0}\} : \mathbf{x} \cdot \mathbf{c} = \tau\Big]\hspace{0.1cm}
\begin{cases}
     \ge 1 - e^{-\Omega(n)} & \text{if~} M  \le |C|^{(1-\epsilon)n} \\[0.3ex]
     \le {|C|^n}\big/{M}     & \text{if~} M \geq |C|^n.
\end{cases}
\]    
\end{theorem}

\begin{lemma}\label{lem:lambda1lowerb}
Suppose $M\ge 4^n$ and fix integers $\alpha,q$ such that
\[
\alpha>\Bigl\lfloor \tfrac14 M^{1/n}\Bigr\rfloor
\qquad\text{and}\qquad
q>\Bigl\lfloor \tfrac14 M^{1/n}\Bigr\rfloor\cdot\sum_{i=1}^n|x_i|.
\]
For $\mathbf{x}\sim[0:M-1]^n$,
\[
\Pr\!\left[\lambda_1^{(\infty)}(\mathcal{L}_{\alpha,q})> \tfrac14 M^{1/n}\right]
\ge 1-e^{-\Omega(n)}.
\]
\end{lemma}

\begin{proof}
Let $d_0:=\left\lfloor \tfrac14 M^{1/n}\right\rfloor\geq 1$.
If $\lambda_1^{(\infty)}(\mathcal{L}_{\alpha,q})\leq\tfrac14 M^{1/n}$, then
$\mathcal{L}_{\alpha,q}$ contains a nonzero vector $\mathbf{v}=(v_0,\dots,v_n)$
with $\|\mathbf{v}\|_\infty\le d_0$.
By Lemma~\ref{lem:connsolandshortvec} for $\tau=0$,
this implies the existence of a nonzero $\mathbf{c}\in[-d_0,d_0]^n$ such that
$\mathbf{x}\cdot\mathbf{c}=0$.

Therefore,
\[
\Pr\!\left[\lambda_1^{(\infty)}(\mathcal{L}_{\alpha,q})\leq\tfrac14 M^{1/n}\right]
\le
\Pr_{\mathbf{x}}\!\left[\exists\,\mathbf{c}\in[-d_0,d_0]^n\setminus\{\mathbf{0}\}:
\mathbf{x}\cdot\mathbf{c}=0\right].
\]
Using the elementary union bound,
\[
\Pr_{\mathbf{x}}\!\left[\exists\,\mathbf{c}\in[-d_0,d_0]^n\setminus\{\mathbf{0}\}:
\mathbf{x}\cdot\mathbf{c}=0\right]
\le
\frac{(2d_0+1)^n}{M}.
\]
Now $2d_0+1\le \tfrac12 M^{1/n}+1$. As we have $M^{1/n}\ge 4$, then
$\tfrac12 M^{1/n}+1\le \tfrac34 M^{1/n}$ and hence
\[
\frac{(2d_0+1)^n}{M}\le \left(\tfrac34\right)^n=e^{-\Omega(n)}.
\]
Therefore, the failure probability is $e^{-\Omega(n)}$, proving the lemma.
\end{proof}

\begin{lemma}
\label{lem:disupperb}
Suppose $M=2^{O(n)}$, let \(\mathbf{x}\sim[0:M-1]^n\) and let \(\tau\in\mathbb{Z}\) satisfy
\(|\tau|=o(nM)\). 
Set \(\mathbf{t}=(\alpha\tau,0,\dots,0)\), for arbitrary integers \(\alpha,q\) and any constant $\varepsilon>0$, then
\[
\mathop{{\bf Pr}}\left[\mathrm{dist}_{\infty}(\mathbf{t}, \mathcal{L}_{\alpha,q})\le \left(\frac{1}{2}+\varepsilon\right) M^{1/n}+1\right]\ge 1-e^{-\Omega(n)}.
\]

\end{lemma}

\begin{proof}
Let \(d_1=\lceil\left(\frac{1}{2}+\varepsilon\right) M^{1/n}\rceil\) and \(C_1=\{-d_1,\dots,d_1\}\). Since $M=2^{O(n)}$, then $d_1$ (so as $|C_1|$) is still a constant. Set
\[
\varepsilon':=\log_{|C_1|}(1+2\varepsilon)>0,
\]
then
\[
|C_1|^{(1-\varepsilon')n}=\frac{|C_1|^n}{(1+2\varepsilon)^n}\ge
\frac{(2(\frac{1}{2}+\varepsilon) M^{1/n})^n}{(1+2\varepsilon)^n}=M.
\]

Now we could apply Theorem~\ref{thm:solprobC0} (with this \(C_1\) and the given \(\tau\)),
we have
\[
\mathop{{\bf Pr}}_{\mathbf{x}}\Big[\exists\ \mathbf{c}\in C_1^n\setminus\{\mathbf{0}\}:\ \mathbf{x}\cdot \mathbf{c}=\tau\Big]
\ge 1-e^{-\Omega(n)}.
\]
Condition on the event that such \(\mathbf{c}=(c_1,\dots,c_n)\) exists. Then \(|c_i|\le d_1\) for all
\(i\), and \(\sum_{i=1}^n c_i x_i=\tau\).

Recall that \(\mathcal{L}_{\alpha,q}\) is generated by the rows of \(\mathbf{B}_{\alpha,q}\), consider the integer linear combination using coefficients \(z_0=0\) and \(z_i=c_i\):
\[
\mathbf{v}
:= \sum_{i=1}^n c_i(\alpha x_i,0,\dots,1,\dots,0)
= \bigl(\alpha\sum_{i=1}^n c_i x_i,\ c_1,\dots,c_n\bigr)
= (\alpha\tau,c_1,\dots,c_n).
\]
Hence \(\mathbf{v}\in \mathcal{L}_{\alpha,q}\) for \emph{any} integers \(\alpha,q\).
Moreover,
\[
\|\mathbf{v}-\mathbf{t}\|_\infty
= \|(0,c_1,\dots,c_n)\|_\infty
= \max_{1\le i\le n}|c_i|
\le d_1.
\]
Therefore \(\mathrm{dist}_\infty(\mathbf{t},\mathcal{L}_{\alpha,q})\le \left(\frac{1}{2}+\varepsilon\right) M^{1/n}+1\) holds on this event, and we conclude
\[
\mathop{{\bf Pr}}\left[\mathrm{dist}_{\infty}(\mathbf{t}, \mathcal{L}_{\alpha,q})\le \left(\frac{1}{2}+\varepsilon\right) M^{\frac{1}{n}}+1\right]
\ge 1-e^{-\Omega(n)}.
\]
\end{proof}

From Lemma~\ref{lem:lambda1lowerb} and Lemma~\ref{lem:disupperb} with $\varepsilon=1/4$, we could deduce that with high probability
\begin{align*}
    \mathrm{dist}_{\infty}(\mathbf{t}, \mathcal{L}_{\alpha,q})&\le \frac{3}{4}M^{1/n}+1\\
    &\leq \frac{3}{4}\cdot 4\cdot\lambda_1^{(\infty)}(\mathcal{L}_{\alpha,q})+1\\
    &\leq 4 \lambda_1^{(\infty)}(\mathcal{L}_{\alpha,q}).
\end{align*}
In this setting, we can solve $\text{CVP}_{\infty}$ in single exponential time via Theorem~\ref{thm:cvpanynorm}. However, we must identify and abort when the ratio between $\mathrm{dist}_{\infty}(\mathbf{t}, \mathcal{L}_{\alpha,q})$ and $\lambda_1^{(\infty)}(\mathcal{L}_{\alpha,q})$ is large. We therefore present Algorithm~\ref{alg:gssbyellenum} rather than applying Theorem~\ref{thm:cvpanynorm} directly.

\begin{algorithm}[t]
\caption{Solving Generalized Subset Sum via $\mathrm{Ellipsoid\mbox{-}Enum}$}
\label{alg:gssbyellenum}
\KwIn{Bound $M$, vector $\mathbf{x} \sim [0:M-1]^n$, $C=[-d:d]$ of allowed coefficients, and a target integer $\tau$ with $|\tau|=o(Mn)$.}
\KwOut{A vector set $S$ or $\perp$.}

Construct lattice $\mathcal{L}_{\alpha,q}$ with $\alpha=\max\{d+1,[M^{\frac{1}{n}}]\},q=\alpha\sum_{i=1}^{n}|x_i|+|\tau|+1$\;

Set target vector $\mathbf{t}=(\alpha\tau,0,\dots,0)$\;

Compute $\lambda_1^{(\infty)}(\mathcal{L}_{\alpha,q})$ using Theorem~\ref{thm:svp-infty}\;\label{step:svpinfty}

\If{$\lambda_1^{(\infty)}(\mathcal{L}_{\alpha,q})\leq\frac{1}{4}M^{\frac{1}{n}}$}{
\Return $\perp$\;
}

Call Lemma~\ref{lem:Ellipsoid-Enum} to compute
\[
U \leftarrow \mathcal{L}_{\alpha,q} \cap \bigl(\mathbf{t}+[M^{\frac{1}{n}}]\sqrt{n+1}\,\mathbf{B}_2^{n+1}\bigr),
\]
$S \leftarrow \{\,\mathbf{v}\in U : \|\mathbf{v}-\mathbf{t}\|_\infty \le d\,\}$\;\label{step:callellenum}
    
\Return $S$

\end{algorithm}

\begin{theorem}
\label{thm:cvpenumalg}
For average-case Generalized Subset Sum, given a bound $M=2^{O(n)}$, vector $\mathbf{x} \sim [0:M-1]^n$, $C=[-d:d]$ of allowed coefficients, and a target integer $\tau$ with $|\tau|=o(Mn)$. Then with probability at least $1-e^{-\Omega(n)}$, if there exists a solution,  Algorithm~\ref{alg:gssbyellenum} returns a non-empty set $S$ in deterministic time
\[
\tilde{O}\!\left(\bigl(18\sqrt{2\pi e}\,\bigr)^n\right)\approx \tilde{O}(2^{6.217n}).
\]
\end{theorem}

\begin{proof}
Set $m=n+1,R=\left[ M^{1/n}\right]$ and use \(\mathbf t=(\alpha\tau,0,\dots,0)\) as the target in Algorithm~\ref{alg:gssbyellenum}.

\paragraph{Algorithm outline and soundness.}
The algorithm has two phases.

\smallskip
\noindent\emph{Phase 1 (guarding \(\lambda_1^{(\infty)}\)).}
It first computes \(\lambda_1^{(\infty)}(\mathcal L_{\alpha,q})\) exactly via Theorem~\ref{thm:svp-infty}.
If \(\lambda_1^{(\infty)}(\mathcal L_{\alpha,q})\le\tfrac14 M^{1/n}\), it immediately returns \(\perp\).
This guard is used to prevent the subsequent enumeration radius from being an excessively large
multiple of \(\lambda_1^{(\infty)}(\mathcal L_{\alpha,q})\), which would make the worst-case output size
\(|\mathcal{L} \cap \bigl(\mathbf{t}+R\sqrt{m}\,\mathbf{B}_2^{m}\bigr)|\) too large.

\smallskip
\noindent\emph{Phase 2 (enumeration and filtering).}
Otherwise it calls Lemma~\ref{lem:Ellipsoid-Enum} to enumerate
\[
U \leftarrow \mathcal{L} \cap \bigl(\mathbf{t}+R\sqrt{m}\,\mathbf{B}_2^{m}\bigr),
\qquad\text{where } R:=\left[ M^{1/n}\right],
\]
and then outputs
\[
S \leftarrow \{\,\mathbf{v}\in U : \|\mathbf{v}-\mathbf{t}\|_\infty \le d\,\}.
\]

We now show \emph{soundness}: any \(\mathbf v\in S\) encodes a valid GSS solution for \(C=[-d:d]\).
Indeed, by construction \(\mathbf v\in\mathcal L_{\alpha,q}\) and \(\|\mathbf v-\mathbf t\|_\infty\le d\).
The algorithm sets
\[
\alpha=\max\{d+1,[ M^{1/n}]\}>d,
\qquad
q=\alpha\sum_{i=1}^n|x_i|+|\tau|+1>d\sum_{i=1}^n|x_i|+|\tau|,
\]
so the hypotheses of Lemma~\ref{lem:connsolclosevec} hold.
Therefore, Lemma~\ref{lem:connsolclosevec} implies that every $\mathbf{v}\in S$ must be of the form
$\mathbf{v}=(\alpha\tau,c_1,\dots,c_n)$, where $|c_i|\le d$ and $\sum_{i=1}^n c_i x_i=\tau$.
Hence, the algorithm never outputs an invalid solution, failure can only occur when a solution exists but the algorithm returns $\perp$ or outputs the empty set.

\paragraph{Failure probability.}
Assume there exists \(\mathbf c\in[-d:d]^n\) such that \(\mathbf x\cdot\mathbf c=\tau\).
Then Lemma~\ref{lem:connsolclosevec} yields
\[
(\alpha\tau,c_1,\dots,c_n)\in \mathcal L_{\alpha,q}\cap(\mathbf t+d\mathbf{B}_\infty^{m}),
\]
so we have $\mathrm{dist}_\infty(\mathbf t, \mathcal L_{\alpha,q})\le d$.

Under this assumption, Algorithm~\ref{alg:gssbyellenum} can fail only in the following cases.

\smallskip
\noindent\emph{(A) The \(\lambda_1^{(\infty)}\) guard triggers.}
That is, \(\lambda_1^{(\infty)}(\mathcal L_{\alpha,q})\le\tfrac14 M^{1/n}\).
By Lemma~\ref{lem:lambda1lowerb} (and our choice of \(\alpha,q\), which satisfies its premises),
\[
\Pr_{\mathbf x}\!\left[\lambda_1^{(\infty)}(\mathcal L_{\alpha,q})\le\tfrac14 M^{1/n}\right]\le e^{-\Omega(n)}.
\]
This event has exponentially small probability and does not affect the overall success probability.

\smallskip
\noindent\emph{(B) Enumeration radius is too small to capture any close vector.}
The algorithm enumerates within \(\mathbf t+R\sqrt m\,\mathbf{B}_2^{m}\), and hence contains
\(\mathbf t+R\mathbf{B}_\infty^{m}\) since \(\mathbf{B}_\infty^{m}\subseteq \sqrt m\,\mathbf{B}_2^{m}\).
If \(d\le R\) then \(\mathbf t+d\mathbf{B}_\infty^{m}\subseteq \mathbf t+R\mathbf{B}_\infty^{m}\), so any solution vector
is certainly enumerated and this failure mode cannot happen.
Thus the only way the algorithm can miss all valid vectors is when \(d>R\) \emph{and} still
\(\mathrm{dist}_\infty(\mathbf t, \mathcal L_{\alpha,q})>R\).

However, Lemma~\ref{lem:disupperb} states that for any constant \(\varepsilon>0\),
\[
\Pr_{\mathbf x}\!\left[\mathrm{dist}_\infty(\mathbf t, \mathcal L_{\alpha,q})\le
\left(\tfrac12+\varepsilon\right)M^{1/n}+1\right]\ge 1-e^{-\Omega(n)}.
\]
Since $M\ge 4^n$, then \((\tfrac12+\varepsilon)M^{1/n}+1\le [M^{1/n}]= R\) holds for suitable $\varepsilon$, so on this
high-probability event we in fact have \(\mathrm{dist}_\infty(\mathbf t, \mathcal L_{\alpha,q})\le R\),
and hence the enumeration must include at least one vector in \(\mathbf t+R\mathbf{B}_\infty^{m}\),
which after filtering yields \(S\neq\emptyset\).
Therefore,
\[
\Pr_{\mathbf x}\!\left[\mathrm{dist}_\infty(\mathbf t, \mathcal L_{\alpha,q})>R\right]\le e^{-\Omega(n)}.
\]

Combining (A) and (B), conditioned on the existence of a solution, the probability that
Algorithm~\ref{alg:gssbyellenum} returns \(\perp\) or outputs \(S=\emptyset\) is at most \(e^{-\Omega(n)}\).
Equivalently, with probability at least \(1-e^{-\Omega(n)}\), it outputs a non-empty set \(S\).

\paragraph{Running time.}
The total running time is the sum of:

\smallskip
\noindent\emph{(1) Computing \(\lambda_1^{(\infty)}(\mathcal L_{\alpha,q})\).}
By Theorem~\ref{thm:svp-infty} in dimension \(m\), this costs
\(\tilde O\bigl((6\sqrt{2\pi e})^{m}\bigr)\). (If the algorithm exits early here, the total time is even smaller.)

\smallskip
\noindent\emph{(2) One call to Lemma~\ref{lem:Ellipsoid-Enum}.}
Lemma~\ref{lem:Ellipsoid-Enum} gives time
\[
\tilde O\!\left(2^{2m}+2^{m}\cdot |U|\right),
\quad\text{where}\quad
U=\mathcal L_{\alpha,q}\cap(\mathbf t+R\sqrt m\,\mathbf{B}_2^{m}).
\]
We bound \(|U|\) via the standard translate-counting routine used in Theorem~\ref{thm:svp-infty}:
\[
|U|\le G(R\sqrt m\,\mathbf{B}_2^{m},\mathcal L_{\alpha,q})
\le N(\sqrt m\,\mathbf{B}_2^{m},\mathbf{B}_\infty^{m})\cdot G(R\mathbf{B}_\infty^{m},\mathcal L_{\alpha,q}).
\]
As in the proof of Theorem~\ref{thm:svp-infty}, the Rogers--Zong volumetric covering bound yields
\[
N(\sqrt m\,\mathbf{B}_2^{m},\mathbf{B}_\infty^{m})=\tilde O\!\left((\sqrt{2\pi e})^{m}\right).
\]
Moreover, on the branch where we do enumerate, we have
\(\lambda_1^{(\infty)}(\mathcal L_{\alpha,q})\ge \tfrac14 M^{1/n}\), and \(R=[ M^{1/n}]\).
By Lemma 4.3 in \cite{dadush2011enumerative},
\[
G(R\mathbf{B}_\infty^{m},\mathcal L_{\alpha,q})\le \left(1+\frac{2R}{\lambda_1^{(\infty)}(\mathcal L_{\alpha,q})}\right)^{m}\le 9^{m}.
\]
Therefore,
\[
|U|\le \tilde O\!\left((\sqrt{2\pi e})^{m}\cdot 9^{m}\right)
=\tilde O\!\left((9\sqrt{2\pi e})^{m}\right).
\]
Plugging into Lemma~\ref{lem:Ellipsoid-Enum},
\[
\tilde O\!\left(2^{2m}+2^{m}\cdot |U|\right)
\le
\tilde O\!\left(2^{2m}+2^{m}(9\sqrt{2\pi e})^{m}\right)
=
\tilde O\!\left((18\sqrt{2\pi e})^{m}\right).
\]
Since \(m=n+1\), this equals \(\tilde O\bigl((18\sqrt{2\pi e})^{n}\bigr)\). The last filtering step only adds polynomial complexity, which does not affect the single-exponential leading term.

Combining all above proves that, with probability at least \(1-e^{-\Omega(n)}\) (conditioned on the existence of a solution),
Algorithm~\ref{alg:gssbyellenum} returns a non-empty set \(S\) and runs in deterministic time
\[
\tilde{O}\!\left(\bigl(18\sqrt{2\pi e}\bigr)^n\right)\approx \tilde{O}(2^{6.217n}).
\qedhere
\]
\end{proof}

\begin{remark}
We have not tried particularly hard to optimize the value of $6.217$, and it seems there is an opportunity to improve it to $4.632$ as the constant in SVP.      
\end{remark}

\begin{theorem}
\label{thm:averagecaseinclude0}
Fixed any $d \in \mathbb{N}$ and let $C = [-d:d]$, there is a deterministic algorithm for average-case Generalized Subset Sum Problems in running time
\[
\tilde{O}\!\left(\bigl(18\sqrt{2\pi e}\,\bigr)^n\right)\approx \tilde{O}(2^{6.217n}).
\]
Given any $M$ and $\tau$ with $|\tau|=o(nM)$, the algorithm succeeds on $(M,\tau,\mathbf{x})$ with probability at least $1-e^{-\Omega(n)}$, over the draw of $\mathbf{x}\sim [0:M-1]^n$.
\end{theorem}

\begin{proof}
We give a deterministic algorithm and analyze its success probability over
$\mathbf{x}\sim[0:M-1]^n$.

If $M>(2d+1)^{2n}$, report no solution.
Otherwise (so $M=2^{O(n)}$), run Algorithm~\ref{alg:gssbyellenum}.
If it outputs a non-empty set $S$, pick any $\mathbf{v}=(v_0,v_1,\dots,v_n)\in S$
and return $\mathbf{c}=(v_1,\dots,v_n)$.
If it outputs $S=\emptyset$, report no solution.

\paragraph{Correctness.}
By construction of $S$ in Algorithm~\ref{alg:gssbyellenum}, every
$\mathbf{v}\in S$ satisfies $\mathbf{v}\in\mathcal{L}_{\alpha,q}$ and
$\|\mathbf{v}-\mathbf{t}\|_\infty\le d$, where $\mathbf{t}=(\alpha\tau,0,\dots,0)$.
Lemma~\ref{lem:connsolclosevec} then implies that $\mathbf{c}=(v_1,\dots,v_n)$
lies in $[-d:d]^n$ and satisfies $\mathbf{x}\cdot\mathbf{c}=\tau$.
Hence whenever the algorithm outputs a vector, it is a correct solution.

\paragraph{Success probability.}
We bound the probability that the algorithm reports no solution while a solution exists.

\smallskip
\noindent\emph{Case 1: $M>(2d+1)^{2n}$.}
By Theorem~\ref{thm:solprobC0}, the probability over $\mathbf{x}$ that there exists
$\mathbf{c}\in[-d:d]^n$ with $\mathbf{x}\cdot\mathbf{c}=\tau$ is at most
\[
\Pr_{\mathbf{x}}[\exists\,\mathbf{c}\in C^n:\ \mathbf{x}\cdot\mathbf{c}=\tau]
\le \frac{|C|^n}{M}=e^{-\Omega(n)}.
\]
Therefore reporting no solution is correct with probability at least $1-e^{-\Omega(n)}$.

\smallskip
\noindent\emph{Case 2: $M\le (2d+1)^{2n}$ (thus $M=2^{O(n)}$).}
In this regime we run Algorithm~\ref{alg:gssbyellenum}.
If a solution exists, by Theorem~\ref{thm:cvpenumalg}, Algorithm~\ref{alg:gssbyellenum} returns $\perp$ or an empty set in probability at most $e^{-\Omega(n)}$.

Hence, with probability at least $1-e^{-\Omega(n)}$, the algorithm returns a correct answer. The running time of the algorithm is bounded by the time of Algorithm~\ref{alg:gssbyellenum}, which is 
\[
\tilde{O}\!\left(\bigl(18\sqrt{2\pi e}\,\bigr)^n\right)\approx \tilde{O}(2^{6.217n}).\qedhere
\]
\end{proof}

\subsection{Generalized Subset Sum with $C=[\pm d]$}
\label{subsec:pmd}

We resolve the worst-case problem in~\cref{subsecgss:pmd} by reducing to approximate \(\mathrm{CVP}_\infty\) algorithm. In the average setting, we adopt a different approach.

We could also prove a bounded distance promise holds in high probability, then we enumerate all lattice vectors closed to the target vector using Lemma~\ref{lem:Ellipsoid-Enum}. Finally we filter out all the vectors that met the conditions.

\begin{theorem}[Theorem 3 in~\cite{chen2022average}]
\label{thm:solprobC}
Let $C=[\pm d]$ for some fixed $d>1$, and fix any constant $\epsilon >0$.
For $\mathbf{x} \sim [0:M-1]^n$ and any integer $\tau$  satisfying $|\tau|=o(Mn)$, we have
\[
\mathop{{\bf Pr}}_{\mathbf{x}}\Big[\exists\hspace{0.05cm} \mathbf{c} \in C^n\setminus\{\mathbf{0}\} : \mathbf{x} \cdot \mathbf{c} = \tau\Big]\hspace{0.1cm}
\begin{cases}
     \ge 1 - o_n(1) & \text{if~} M  \le |C|^{(1-\epsilon)n} \\[0.3ex]
     \le {|C|^n}\big/{M}     & \text{if~} M \geq |C|^n.
\end{cases}
\]    
\end{theorem}

The above Theorem states that in the same condition as in Theorem~\ref{thm:solprobC0}, and the only difference is that probability changes from $1-e^{-\Omega(n)}$ to $1 - o_n(1)$.

\begin{theorem}
\label{thm:averagecaseexclude0}
Fixed any $d >1$ and let $C = [\pm d]$, there is a deterministic algorithm for average-case Generalized Subset Sum Problems in running time
\[
\tilde{O}\!\left(\bigl(18\sqrt{2\pi e}\,\bigr)^n\right)\approx \tilde{O}(2^{6.217n}).
\]
Given any $M$ and $\tau$ with $|\tau|=o(nM)$, the algorithm succeeds on $(M,\tau,\mathbf{x})$ with probability at least $1-o_n(1)$, over the draw of $\mathbf{x}\sim [0:M-1]^n$.   
\end{theorem}

\begin{proof}
Recall that in our notation \(C=[\pm d]=[-d:d]\setminus\{0\}\), and \(|C|=2d\) is a constant.

\paragraph{Algorithm.}
Given \((M,\tau,\mathbf x)\), the algorithm proceeds as follows.

\begin{enumerate}
\item If \(M>|C|^{2n}=(2d)^{2n}\), output ``no solution''.
\item Otherwise (hence \(M=2^{O(n)}\)), run Algorithm~\ref{alg:gssbyellenum} (with the same lattice construction)
and obtain either \(\perp\) or a set \(S\).
If \(S=\emptyset\), output ``no solution''.
\item If \(S\neq\emptyset\), perform one additional filtering step:
\[
S' \leftarrow \bigl\{\mathbf v=(v_0,v_1,\dots,v_n)\in S:\ v_i\neq 0\ \text{for all } i=1,\dots,n\bigr\}.
\]
If \(S'\neq\emptyset\), output any \(\mathbf c=(v_1,\dots,v_n)\) from some \(\mathbf v\in S'\);
otherwise output ``no solution''.
\end{enumerate}

\paragraph{Correctness.}
Whenever the algorithm outputs some \(\mathbf c=(c_1,\dots,c_n)\), it comes from a vector
\(\mathbf v=(\alpha\tau,c_1,\dots,c_n)\in S\subseteq \mathcal L_{\alpha,q}\) satisfying
\(\|\mathbf v-\mathbf t\|_\infty\le d\).
By Lemma~\ref{lem:connsolclosevec}, this implies \(\sum_{i=1}^n c_ix_i=\tau\) and \(|c_i|\le d\).
Moreover, since \(\mathbf v\in S'\), we also have \(c_i\neq 0\) for all \(i\), hence \(\mathbf c\in C^n\).
Therefore the algorithm never outputs an invalid solution.

\paragraph{Success probability over \(\mathbf x\sim[0:M-1]^n\).}
We bound the probability that the algorithm outputs ``no solution'' despite the existence of a solution in \(C^n\).

\smallskip
\noindent\emph{Case 1: \(M>(2d)^{2n}\).}
By Theorem~\ref{thm:solprobC} applied to this set \(C\),
\[
\Pr_{\mathbf x}\bigl[\exists\,\mathbf c\in C^n\setminus\{\mathbf 0\}:\ \mathbf x\cdot\mathbf c=\tau\bigr]
\le \frac{|C|^n}{M}
< \frac{(2d)^n}{(2d)^{2n}}
=(2d)^{-n}
=e^{-\Omega(n)}.
\]
Hence, in this regime, reporting ``no solution'' is correct with probability at least \(1-e^{-\Omega(n)}\).

\smallskip
\noindent\emph{Case 2: \(M\le (2d)^{2n}\) (thus \(M=2^{O(n)}\)).}
Assume there exists \(\mathbf c\in C^n\) with \(\mathbf x\cdot\mathbf c=\tau\). Therefore the corresponding lattice vector \((\alpha\tau,c_1,\dots,c_n)\) lies in
\(\mathcal L_{\alpha,q}\cap(\mathbf t+d\mathbf{B}_\infty^{n+1})\), and all its coordinates \(c_i\) are nonzero.
Like the proof of Lemma~\ref{lem:disupperb}, by replacing Theorem~\ref{thm:solprobC0} with Theorem~\ref{thm:solprobC}, we could obtain that 
\[
(\alpha\tau,\mathbf{c})\in \bigl(\mathbf{t}+[ M^{\frac{1}{n}}]\sqrt{n+1}\,\mathbf{B}_2^{n+1}\bigr)
\]
holds with probability at least \(1-o_n(1)\).

Thus Algorithm~\ref{alg:gssbyellenum} outputs a non-empty set \(S\) that contains this vector with probability at least \(1-o_n(1)\). Since \(c_i\ne 0\), this vector would retain after the filtering step and $S'$ is a non-empty set
with probability at least \(1-o_n(1)\).

Combining both cases, the algorithm is correct on input \((M,\tau,\mathbf x)\) with probability at least
\(1-o_n(1)\) over \(\mathbf x\sim[0:M-1]^n\).

\paragraph{Running time.}
The additional filtering step takes time \(\operatorname{poly}(n)\cdot |S|\le \operatorname{poly}(n)\cdot |U|\)
and therefore does not affect the single-exponential leading term.
The running time is dominated by Algorithm~\ref{alg:gssbyellenum}, hence it is
\[
\tilde{O}\!\left(\bigl(18\sqrt{2\pi e}\bigr)^n\right)\approx \tilde{O}(2^{6.217n}).\qedhere
\]
\end{proof}

\paragraph{Proof of \cref{thm:mainaveragecase}.}
For the case when $C=[-d:d]$, applying \cref{thm:averagecaseinclude0}. For the case when $C=[\pm d]$, we use brute-force enumeration for $d=1$ and $C=\{-1,1\}$, otherwise, we apply \cref{thm:averagecaseexclude0}.\qed

% \section{Conclusion}
% \label{sec:conclusion}
% \input{tex_files/conclusion}
\bibliography{mybib}

\appendix

\section{Proof of~\cref{pro:approxcvp}}
\label{sec:appendixapproxcvp}
\begin{proof}
We combine the deterministic $2$-approximation for $\mathrm{CVP}_\infty$ obtained from Dadush and Kun~\cite{dadush2013lattice} with the covering and boosting framework of Eisenbrand, H\"ahnle, and Niemeier~\cite{eisenbrand2011covering}.

\smallskip
\noindent
\textbf{Step 1: a deterministic $2$-approximation oracle.}
Dadush and Kun (Theorem~1.1 \cite{dadush2013lattice}) give a deterministic algorithm that, for any norm and any $0<\delta\le 1$, computes a $(1+\delta)$-approximate closest vector in time
\[
2^{O(n)}\Bigl(1+\frac{1}{\delta}\Bigr)^{\!n}. 
\]
The $\ell_\infty$ norm is symmetric, so the algorithm applies directly.  Setting $\delta = 1$ yields a deterministic $2$-approximation algorithm for $\mathrm{CVP}_\infty$, which we denote by $\mathcal{A}_2$; its running time is $2^{O(n)}$.

\smallskip
\noindent
\textbf{Step 2: $2$-gap oracle from $\mathcal{A}_2$.}
Given a lattice $\mathcal{L}$, a target $\mathbf{t}$ and a distance $D>0$, we scale the problem by $1/D$ so that $D=1$.  One call to $\mathcal{A}_2$ on $(\mathcal{L}, \mathbf{t})$ returns a vector $\mathbf{v}$ with $\|\mathbf{v}-\mathbf{t}\|_\infty\le 2\cdot\mathrm{dist}_{\infty}(\mathbf{t}, \mathcal{L})$.  
\begin{itemize}
    \item If $\|\mathbf{v}-\mathbf{t}\|_\infty\le 2$, we have found a lattice point in $\mathbf{t}+[-1,1]^n$ (after undoing the scaling).
    \item If $\|\mathbf{v}-\mathbf{t}\|_\infty>2$, then necessarily $\mathrm{dist}_{\infty}(\mathbf{t}, \mathcal{L})>1$, which certifies that $\mathbf{t}+[-1,1]^n$ contains no lattice point.
\end{itemize}
Thus a single call to $\mathcal{A}_2$ implements a deterministic $2$-gap oracle as required in~\cite{eisenbrand2011covering}.

\smallskip
\noindent
\textbf{Step 3: boosting to $(1+\varepsilon)$-approximation.}
Eisenbrand~\textit{et~al.}~\cite{eisenbrand2011covering} show how to boost any $2$-gap oracle to a $(1+\varepsilon)$-approximation using two completely deterministic subroutines.

\begin{enumerate}
    \item \emph{From $2$-gap to $(1+\varepsilon)$-gap.}  
    Theorem~3.2 of~\cite{eisenbrand2011covering} reduces the $(1+\varepsilon)$-gap decision problem to at most
    \[
    2^{n}\Bigl(2+\log\frac{1}{\varepsilon}\Bigr)^{n} = 2^{O(n)}\Bigl(\log\frac{1}{\varepsilon}\Bigr)^{O(n)}
    \]
    calls of the $2$-gap oracle.  The reduction is based on covering the cube $[-1+\delta,1-\delta]^n$ (with $\delta = \varepsilon/(1+\varepsilon)$) by axis-parallel parallelepipeds that, when scaled by $2$, stay inside $[-1,1]^n$ (Theorem~2.6).  The covering is built by a coordinate-wise binary partition of the cube and is therefore deterministic.  Each parallelepiped induces one instance of the $2$-gap problem, and the encoding length of each instance stays polynomial in $n$, $\log(1/\varepsilon)$ and the original input size $b$.

    \item \emph{From $(1+\varepsilon)$-gap to $(1+\varepsilon)$-approximation.}
    Theorem~3.4 of~\cite{eisenbrand2011covering} gives a deterministic binary search that, given a $(1+\delta)$-gap oracle with $\delta = \min\{\varepsilon/5,1/2\}$, outputs a $(1+\varepsilon)$-approximate closest vector.  The search makes
    \[
    O\bigl(\log b + \log n + \log(1/\varepsilon)\bigr)
    \]
    calls to the $(1+\delta)$-gap oracle.
\end{enumerate}

\smallskip
\noindent
\textbf{Total complexity.}
We compose the reductions deterministically:
\begin{itemize}
    \item The $2$-gap oracle is implemented by $\mathcal{A}_2$.
    \item Solving the $(1+\varepsilon)$-gap problem uses $\ell_1 = 2^{O(n)}(\log\frac{1}{\varepsilon})^{O(n)}$ calls to $\mathcal{A}_2$.
    \item The final binary search uses $\ell_2 = O(\log b + \log n + \log\frac{1}{\varepsilon})$ calls to the $(1+\varepsilon)$-gap solver.
\end{itemize}
Hence the overall algorithm calls $\mathcal{A}_2$ at most $\ell_1 \cdot \ell_2 = 2^{O(n)}(\log\frac{1}{\varepsilon})^{O(n)}$ times.  Each call takes $2^{O(n)}\mathrm{poly}(b)$ time and the size of every constructed instance remains $\mathrm{poly}(n,b,\log\frac{1}{\varepsilon})$, so the total running time is
\[
2^{O(n)}\Bigl(\log\frac{1}{\varepsilon}\Bigr)^{O(n)}\mathrm{poly}(b) \;=\; 2^{O(n)}\Bigl(\log\frac{1}{\varepsilon}\Bigr)^{O(n)} .
\]
All components are deterministic, thus the complete procedure is a deterministic $(1+\varepsilon)$-approximation for $\mathrm{CVP}_\infty$.
\end{proof}
\end{document}